\documentclass[11pt]{article}

\def\Comment#1{{\bf \textsl{$\langle\!\langle$#1\/$\rangle\!\rangle$}}}

\advance \topmargin by -1.0in
\advance \oddsidemargin by -0.8in

\newcommand{\headers}[3]{
\newpage\setcounter{page}{1}
\def\@oddhead{$\underline{\hbox to\textwidth{%
\textbf{\rlap{#1}\phantom{hj}\hfill #2 \hfill \llap{#3}}}}$}
\def\@oddfoot{\hfill\thepage\hfill}}

\usepackage{amsfonts,amsmath,amssymb}
\usepackage{enumerate, graphicx}
\usepackage{xspace}
\usepackage{boxedminipage}
\usepackage{url}
\usepackage{float}
\usepackage{paralist}
\usepackage{accents}
\usepackage{setspace}
\usepackage{hyperref}
\usepackage{mathrsfs}
\usepackage[cm]{fullpage}%
\usepackage{paralist}%
\usepackage{scalerel}%
\usepackage[amsmath,thmmarks]{ntheorem}%
\usepackage{graphicx}
\usepackage[labelfont=bf]{caption}
\usepackage{subcaption}
\usepackage{colortbl}
\usepackage{hhline}
\usepackage{tikz}
\usepackage{multirow}
\usepackage{hhline}
\usepackage{hhline}
\usepackage{soul,footnote}

\usepackage{wrapfig}
\usepackage{color}
\usepackage{cite}
\usepackage{booktabs}
\makesavenoteenv{tabular}
\makesavenoteenv{table}

\usepackage[margin=1in]{geometry}

\usepackage{salgorithm}

\usepackage{mleftright}%

\usepackage{hyperref}%
\hypersetup{%
   breaklinks,%
   ocgcolorlinks,
   colorlinks=true,%
   linkcolor=[rgb]{0.45,0.0,0.0},%
   citecolor=[rgb]{0,0,0.45}
}

\numberwithin{figure}{section}
\numberwithin{table}{section}
\numberwithin{equation}{section}%

\newcommand{\HLinkShort}[2]{\hyperref[#2]{#1\ref*{#2}}}
\newcommand{\HLink}[2]{\hyperref[#2]{#1~\ref*{#2}}}
\newcommand{\HLinkPage}[2]{\hyperref[#2]{#1~\ref*{#2}%
      $_\text{p\pageref{#2}}$}}
\newcommand{\HLinkPageOnly}[1]{\hyperref[#1]{Page~\refpage*{#1}%
      $_\text{p\pageref{#1}}$}}

\newcommand{\HLinkSuffix}[3]{\hyperref[#2]{#1\ref*{#2}{#3}}}
\newcommand{\HLinkPageSuffix}[3]{\hyperref[#2]{#1\ref*{#2}%
      #3$_\text{p\pageref{#2}}$}}

\providecommand{\eqlab}[1]{}%
\renewcommand{\eqlab}[1]{\label{equation:#1}}

\theoremstyle{plain}%
\newtheorem{theorem}{Theorem}[section]

\newtheorem{lemma}[theorem]{Lemma}

\newtheorem{corollary}[theorem]{Corollary}

\newtheorem{proposition}[theorem]{Proposition}
\newtheorem{prop}[theorem]{Proposition}

\theoremstyle{plain}%
\theoremheaderfont{\sf} \theorembodyfont{\upshape}%
\newtheorem*{remark:unnumbered}[theorem]{Remark}%
\newtheorem{remark}[theorem]{Remark}%
\newtheorem{definition}[theorem]{Definition}

\theoremheaderfont{\em}%
\theorembodyfont{\upshape}%
\theoremstyle{nonumberplain}%
\theoremseparator{}%
\theoremsymbol{\myqedsymbol}%

\newenvironment{proof}{\noindent{\bf Proof:}}%
{\hspace*{\fill}$\Box$\par \vspace{4mm}}

\newcommand{\myqedsymbol}{$\square$}

\newenvironment{proofof}[1]{\smallskip\noindent{\bf Proof of #1:}}%
{\hspace*{\fill}$\Box$\par \vspace{3mm}}

\def\compactify{\itemsep=0pt \topsep=0pt \partopsep=0pt \parsep=0pt}
\let\latexusecounter=\usecounter

{\begin{itemize}%
\setlength{\itemsep}{0pt}%
\setlength{\topsep}{0pt}%
\setlength{\partopsep}{0pt}%
\setlength{\parskip}{0pt}}%
{\end{itemize}}

\newcommand{\eps}{\varepsilon}%

\def\script#1{\mathcal{#1}}
\def\opt{\textsc{OPT}}

\def\etal{\text{et al.}\xspace}

\def\sep{\;|\;}
\def\card#1{|#1|}
\def\set#1{\{#1\}}

\def\path{\mathrm{path}}

\newcommand{\pth}[1]{\mleft({#1}\mright)}%

\def\iter{i}

\def\index{\texttt{index}}

\newcommand{\prob}[1]{\textup{\fontencoding{T1}\fontfamily{ppl}\fontseries{m
}\fontshape{n}\selectfont #1}\xspace}

\def\sC{\script{C}}

\def\eps{\varepsilon}

\def\script#1{\mathcal{#1}}
\def\opt{\textsc{OPT}}
\def\etal{\text{et al.}\xspace}

\def\sep{\;|\;}
\def\mA{\script{A}}

\def\mL{\script{L}}
\def\mLB{\script{L}_{\sf b}}
\def\mLR{\script{L}_{\sf r}}
\def\mP{\script{P}}
\def\mPE{\script{P}_{{\sf elem}}}
\def\mPV{\script{P}_{{\sf vc}}}
\def\mC{\script{C}}
\def\mT{\script{T}}
\def\mF{\script{F}}

\def\mG{\script{G}}
\def\mTB{\mT_{\sf b}}
\def\mTR{\mT_{\sf r}}

\def\card#1{|#1|}
\def\set#1{\{#1\}}
\def\bd(#1){\textup{\fontencoding{T1}\fontfamily{ppl}\fontseries{m}\fontshape{n}\selectfont {bd}}\xspace(#1)}

\def\hC{\hat{C}}

\def\hS{\hat{S}}
\def\hT{\hat{T}}
\def\hX{\hat{X}}
\def\hY{\hat{Y}}

\def\NP{\ensuremath{\mathrm{\mathbf{NP}}}}

\newcommand{\algViolated}{\Algorithm{ViolatedBisets}\xspace}

\newcommand{\algCover}{\Algorithm{Cover}\xspace}
\renewcommand{\Algorithm}[1]{{\AlgorithmI{#1}}}
\def\rE{r_{\sf elem}}
\def\fE{f_{\sf elem}}
\def\rV{r_{\sf v}}

\def\hV{h_{\sf v}}

\begin{document}
\title{Node-Weighted Network Design in Planar\\ and Minor-Closed
  Families of Graphs\thanks{This work was partially supported by NSF
    grant CCF-1016684. This work was mostly done when the second and third
    authors were students at University of Illinois. This paper builds
    upon an earlier version with results on edge-connectivity that appeared
    in {\em Proceedings of ICALP, 2012}.}
}

\author{
Chandra Chekuri\thanks{Dept.\ of Computer Science;
	University of Illinois; Urbana; IL 61801, USA; %
      \href{mailto:chekuri@illinois.edu}{chekuri@illinois.edu}.}
\and
Alina Ene\thanks{
      Dept.\ of Computer Science, Boston University; Boston; MA 02215, USA; %
      \href{mailto:aene@bu.edu}{aene@bu.edu}.}
\and
Ali Vakilian\thanks{
      Dept.\ of Computer Science, University of Wisconsin; Madison; WI
      53706, USA; %
      \href{mailto:vakilian@wisc.edu}{vakilian@wisc.edu}.}%
}
\date{}
\maketitle

\thispagestyle{empty}
\begin{abstract}
  We consider {\em node-weighted} survivable network design
  (\prob{SNDP}) in planar graphs and minor-closed families of
  graphs. The input consists of a node-weighted undirected graph
  $G=(V,E)$ and integer connectivity requirements $r(uv)$ for each
  unordered pair of nodes $uv$. The goal is to find a minimum
  weighted subgraph $H$ of $G$ such that $H$ contains $r(uv)$
  disjoint paths between $u$ and $v$ for each node pair $uv$. Three
  versions of the problem are edge-connectivity \prob{SNDP}
  (\prob{EC-SNDP}), element-connectivity \prob{SNDP}
  (\prob{Elem-SNDP}) and vertex-connectivity \prob{SNDP}
  (\prob{VC-SNDP}) depending on whether the paths are required to be
  edge, element or vertex disjoint respectively. Our main result is an
  $O(k)$-approximation algorithm for \prob{EC-SNDP} and
  \prob{Elem-SNDP} when the input graph is planar or more generally if
  it belongs to a proper minor-closed family of graphs; here
  $k=\max_{uv} r(uv)$ is the maximum connectivity requirement. This
  improves upon the $O(k \log n)$-approximation known for node-weighted
  \prob{EC-SNDP} and \prob{Elem-SNDP} in general graphs
  \cite{Nutov10}.  We also obtain an $O(1)$ approximation for
  node-weighted \prob{VC-SNDP} when the connectivity requirements are
  in $\{0,1,2\}$; for higher connectivity our result for
  \prob{Elem-SNDP} can be used in a black-box fashion to obtain a
  logarithmic factor improvement over currently known general graph
  results. Our results are inspired by, and generalize, the work of
  Demaine, Hajiaghayi and Klein \cite{DemaineHK14} who obtained
  constant factor approximations for node-weighted Steiner tree and
  Steiner forest problems in planar graphs and proper minor-closed
  families of graphs via a primal-dual algorithm.
\end{abstract}

\section{Introduction}
\label{sec:intro}
Network design is an important area of discrete optimization with
several practical applications. Moreover, the clean optimization
problems that underpin the applications have led to fundamental
theoretical advances in combinatorial optimization, algorithms and
mathematical programming. In this paper we consider a class of
problems that can be modeled as follows. Given an undirected graph
$G=(V,E)$ find a subgraph $H$ of {\em minimum weight/cost} such that
$H$ satisfies certain desired {\em connectivity} properties.  A common
cost model is to assign a non-negative weight $w(e)$ to each $e \in E$
and the weight of $H$ is simply the total weight of edges in
it. A number of well-studied problems can be cast as special
cases. Examples include polynomial-time solvable problems such as the
minimum weight spanning tree (\prob{MST}) problem when $H$ is required
to connect all the nodes of $G$, and the \NP-hard \prob{Steiner Tree}
problem where $H$ is required to connect only a given subset $S
\subseteq V$ of terminals. A substantial generalization of these
problems is the {\em survivable network design problem} (\prob{SNDP})
which is defined as follows. The input, in addition to $G$, consists
of an integer requirement function $r(uv)$ for each (unordered) pair
of nodes $uv$ in $G$; the goal is to find a minimum-weight subgraph
$H$ such that it contains, for each pair of nodes $uv$, $r(uv)$
disjoint-paths between $u$ and $v$. We obtain two fundamental
variants: if the $r(uv)$ paths for $uv$ are required to be
edge-disjoint it is called edge-connectivity \prob{SNDP} problem
(\prob{EC-SNDP}), and if they are required to be internally
vertex-disjoint the problem is called vertex-connectivity \prob{SNDP}
(\prob{VC-SNDP}). These problems are relevant in designing
fault-tolerant networks. It is not hard to see that \prob{VC-SNDP} is
a generalization of \prob{EC-SNDP}. Moreover \prob{VC-SNDP} is known
to be strictly harder than \prob{EC-SNDP} from an approximation point
of view. A problem of intermediate complexity is the
element-connectivity \prob{SNDP} problem (\prob{Elem-SNDP}): here the
vertex set $V$ is partitioned into \emph{reliable nodes} $R$ and
\emph{non-reliable nodes} $V\setminus R$. The requirements are only
between terminal nodes $T \subseteq R$. The goal is to find a subgraph
$H$ of minimum weight such that each pair of terminals $uv$ has
$r(uv)$ element-disjoint paths, that is, paths that are disjoint in
edges and non-reliable nodes. The problems mentioned so far arise
naturally in concrete applications.  Algorithmic approaches for these
problems are in fact based on solving a larger class of abstract
network design problems such as covering proper and skew-supermodular
cut-requirement functions that we describe formally later.

\medskip
\noindent {\bf Node weights:} The cost of a network is dependent on
the application. In connectivity problems, as we remarked, a common
model is the edge-weight model. A more general problem is defined
when each node $v$ of $G$ has a weight $w(v)$ and the weight of $H$ is
the total weight of the nodes in $H$\footnote{For many problems of
  interest including \prob{Steiner Tree} and \prob{SNDP} the version
  with weights on both edges and nodes can be reduced to the
  version with only node weights; simply sub-divide an edge $e$ by
  placing a new node $v_e$ and place the weight of $e$ on
  $v_e$.}. Node weights are relevant in several applications, in
particular telecommunication networks, where they can model the cost
of setting up routing and switching infrastructures at a given node.
There have also been several recent applications in wireless network
design \cite{Panigrahi11,Nutov12-latin} where the weight function is closely
related to that of node weights. We refer the reader to
\cite{DemaineHK14} for some additional applications of node weights
to network formation games.

The node-weighted versions of network design problems often turn out
to be strictly harder to approximate than their corresponding
edge-weighted versions.  For instance the \prob{Steiner Tree} problem
admits a $1.39$-approximation for edge-weights \cite{ByrkaGRS10};
however, Klein and Ravi \cite{KleinR95} showed via a simple reduction
from the \prob{Set Cover} problem that the node-weighted \prob{Steiner
  Tree} problem on $n$ nodes is hard to approximate to within an
$\Omega(\log n)$-factor unless ${\bf P}={\bf NP}$.
They also described a $(2\log k)$-approximation where $k$ is the
number of terminals. A more dramatic difference emerges for
\prob{SNDP}. While Jain gave a $2$-approximation for \prob{EC-SNDP}
with edge-weights \cite{Jain01}, the best known approximation
guarantee for \prob{EC-SNDP} with node-weights is $O(k \log n)$
\cite{Nutov10} where $k = \max_{uv} r(uv)$ is the maximum connectivity
requirement. Moreover, Nutov \cite{Nutov10} gives evidence, via a
reduction from the \prob{$k$-Densest-Subgraph} problem, that for the
node-weighted problem a dependence on $k$ in the approximation ratio
is necessary.  Table~\ref{table:approx-ratios} summarizes the known
approximation ratios and hardness results for some of the main
problems considered in this paper. One notices that the node-weighted
version of problems have at least a logarithmic factor worse
approximation than the corresponding edge-weighted problem.

\begin{table}[t]
  \centering
  \resizebox{\textwidth}{!}{%
  \renewcommand{\arraystretch}{1.5}
  \begin{tabular}{llllll}
    \toprule
    \multirow{2}{*}{}& \multicolumn{2}{c}{\bf Edge Weighted Graphs} & & \multicolumn{2}{c}{\bf Node Weighted Graphs} \\
    \cline{2-3}\cline{5-6}
    & {\bf General} & {\bf Planar} & & {\bf General} & {\bf Planar} \\ 
    \midrule
    \prob{Steiner Tree} & 1.39 \cite{ByrkaGRS13} & PTAS
                                                   \cite{BorradaileKM09}
                              & & $O(\log n)$ \cite{KleinR95} &  2.4
                                                              \cite{BermanY12}
    \\ 
    \prob{Steiner Forest} & 2 \cite{AgrawalKR95} & PTAS
                                                   \cite{BateniHM11} &
                                                                       & $O(\log n)$ \cite{KleinR95} & 2.4 \cite{BermanY12} \\ 
    \prob{$\{0,1\}$ Proper Functions} & 2 \cite{GoemansW95} & 2
                                                            \cite{GoemansW95}
                                      & & $O(\log n)$
                                              \cite{KleinR95} & 6 \cite{DemaineHK14} \\ 
    \prob{EC-SNDP} & 2 \cite{Jain01} & 2 \cite{Jain01} & & $O(k \log n)$ \cite{Nutov10} & $10k$ \\ 
    \prob{Eelm-SNDP} & 2 \cite{FleischerJW06} & 2 \cite{FleischerJW06}
                                      & & $O(k \log n)$ \cite{Nutov18} & $10k$ \\ 
    \prob{$\{0,1,2\}$ VC-SNDP} & 2
                                                      \cite{FleischerJW06}
                         & 2 \cite{FleischerJW06} & & $O(\log n)$
                                                    \cite{Nutov10} &  $13$ \\ 
    \prob{VC-SNDP} & $O(k^3 \log n)$ \cite{ChuzhoyK12} & $O(k^3 \log
                                                         n)$
                                                         \cite{ChuzhoyK12}
                                      & & $O(k^4 \log^2
                                        n)$ \cite{ChuzhoyK12} & $O(k^4 \log n)$ \\ 
    \bottomrule
  	\end{tabular}
	}
  \caption{Approximation ratios for \prob{SNDP} and related
    problems. The entries with no citation are from this paper. There
    is an $\Omega(\log n)$-hardness for all the node-weighted problems
    in the table for general graphs.}
  \label{table:approx-ratios}
\end{table}

Demaine, Hajiaghayi and Klein \cite{DemaineHK14} considered the
approximability of the node-weighted \prob{Steiner Tree} problem in
\emph{planar} graphs. They were partly motivated by the goal of
overcoming the $\Omega(\log n)$-hardness that holds in general
graphs. They described a primal-dual algorithm that is adapted
from the well-known algorithm for the edge-weighted case
\cite{AgrawalKR95,GoemansW95}, and showed that it gives a
$6$-approximation in planar graphs. Demaine \etal also showed that
their algorithm works for a more general class of $\{0,1\}$-proper functions (first considered by Goemans and Williamson
\cite{GoemansW95}) that includes several other problems such as the
\prob{Steiner Forest} problem.  Their analysis shows that one obtains
a constant factor approximation (the algorithm is the same) for any
proper minor-closed family of graphs where the constant depends on the
family. In addition to their theoretical values, these results have the
potential to be useful in practice; the algorithm is simple and
efficient to implement, and it is reasonable to assume that
real-work networks that arise in several applications are close to
being planar.

\subsection{Our Results}
In this paper we consider node-weighted network design problems in
planar graphs for \emph{higher connectivity}. In particular we consider
\prob{EC-SNDP}, \prob{Elem-SNDP} and \prob{VC-SNDP} and show that the
basic insight in \cite{DemaineHK14} can be built upon to develop improved
approximation algorithms for these more general problems as
well. Although we follow the high-level outline of
\cite{DemaineHK14}, our results require susbtantial technical work.
Our core result is for \prob{Elem-SNDP} which captures \prob{EC-SNDP}
as a special case and can be used in a black box fashion for
\prob{VC-SNDP}.

\begin{theorem}\label{thm:intro-elem-sndp}
  There is a $10k$-approximation algorithm for node-weighted
  \prob{Elem-SNDP} in planar graphs where $k$ is the maximum
  requirement.  Moreover, an $O(k)$ approximation guarantee also holds
  for graphs from a proper minor-closed family of graphs $\mG$ where
  the constant in the approximation factor only depends on the family
  $\mG$.
\end{theorem}

Node-weighted \prob{EC-SNDP} can be reduced easily in an approximation
preserving fashion to node-weighted \prob{Elem-SNDP} by choosing
all nodes in the input graph to be reliable nodes. Thus the preceding
theorem applies to node-weighted \prob{EC-SNDP}.
Chuzhoy and Khanna \cite{ChuzhoyK12} showed a generic reduction of
\prob{VC-SNDP} to \prob{Elem-SNDP} that does not change the underlying graph.
Using Theorem~\ref{thm:intro-elem-sndp} and the reduction in
\cite{ChuzhoyK12} we obtain the following corollary.

\begin{corollary}\label{cor:intro-vc-sndp}
  There is an $O(k^4 \log n)$-approximation for node-weighted \prob{VC-SNDP}
  in proper   minor-closed family of graphs.
\end{corollary}

We also obtain an $O(1)$-approximation algorithm for
node-weighted \prob{VC-SNDP} when the requirements are in the
set $\set{0,1,2}$.

\begin{theorem}\label{thm:0-2-vc-sndp}
  There is a $13$--approximation algorithm for node-weighted
  \prob{VC-SNDP} with $\set{0,1,2}$ connectivity requirements in
  planar graphs and, more generally, an $O(1)$-approximation for graphs from a proper minor-closed
  family of graphs.
\end{theorem}

In summary, as mentioned in Table~\ref{table:approx-ratios}, our
results show that provably better approximation guarantees can be
obtained for node-weighted network design in planar graphs
when compared to the case of general graphs.

\subsection{Overview of Technical Ideas and Contribution}
There are two main algorithmic approaches for \prob{SNDP}. The
first approach is the \emph{augmentation} approach pioneered by Williamson
\etal \cite{WilliamsonGMV95}. In this approach the desired network
is built in $k$ phases; at the end of the first $(\ell-1)$ phases the
connectivity of a pair $uv$ is at least $\min\set{r(uv), \ell-1}$. Thus,
the optimization problem of $\ell$-th phase is to increase the connectivity of certain
pairs by \emph{one}; the advantage is that we need to work with a
$0$-$1$ covering function. In the case of \prob{EC-SNDP}, the augmentation problem
is the problem of covering a skew-supermodular function. A
requirement function $f:2^V \rightarrow \mathbb{Z}_+$ is
skew-supermodular\footnote{This class of functions is also referred
by other names such as uncrossable and weakly supermodular.} if for
any $A, B \subseteq V$
\begin{align*}
f(A) + f(B) \le \max\set{f(A\cap B)+f(A\cup B) ,
f(A\setminus B)+f(B\setminus A)}.
\end{align*}
Williamson \etal \cite{WilliamsonGMV95} showed that a primal-dual
algorithm achieves a $2$-approximation for covering edge-weighted
$0$-$1$ skew-supermodular functions. For the node-weighted variant,
Nutov \cite{Nutov10} gave an $O(\log n)$-approximation. These results
for covering $0$-$1$ skew-supermodular functions when combined with
the augmentation give respectively $2k$ and
$O(k \log n)$-approximation for the edge-weighted and node-weighted
\prob{EC-SNDP}\footnote{The approximation for the edge-weighted
  version can be improved to $2H_k$ by doing the augmentation in the
  reverse order \cite{GoemansGPSTW94}.}.  For solving edge-weighted
\prob{Elem-SNDP} in the augmentation framework,~\cite{JainMVW02} works
with {\it skew-bisupermodular} functions which are more involved
although the achieved approximation ratios are similar.  The second
approach for \prob{SNDP} is the powerful \emph{iterative rounding}
technique pioneered by Jain which led to a $2$-approximation for
\prob{EC-SNDP} \cite{Jain01} and \prob{Elem-SNDP}
\cite{FleischerJW06}.  Iterative rounding does not quite apply to
node-weighted problems for a variety of technical reasons as well as
known hardness of approximation results.  For this reason the main
approach for attacking node-weighted \prob{SNDP} problems has been the
augmentation approach.

In this paper, we follow the augmentation approach for node-weighted
\prob{SNDP} problems. Demaine \etal adapted the primal-dual algorithm
for edge-weighted $0$-$1$ proper functions to the node-weighted
case. The novel technical ingredient in their analysis is to
understand the properties of {\em node-minimal} feasible solutions
instead of edge-minimal feasible solutions. The analysis crucially
relies on the average degree of a planar graph being constant.  For
the most part, problems captured by $0$-$1$ proper functions are very
similar to the \prob{Steiner Forest} problem, a canonical problem in
this class.  In this setting it is possible to visualize and
understand node-minimal solutions through connected components and
basic reachability properties. In the augmentation approach for
higher-connectivity, as we remarked, the problem in each phase is no
longer a proper function but belongs to the richer class of
skew-supermodular functions. The primal-dual analysis for this class
of functions is more subtle and abstract and proceeds via uncrossing
arguments and laminar witness families~\cite{WilliamsonGMV95}.

In a previous conference version of this work~\cite{ChekuriEV12}, we
considered node-weighted \prob{EC-SNDP} (and more general problems) in
planar and minor-closed families of graphs. Based on properties of
node-minimal feasible solutions for $0$-$1$ skew-supermodular families
we obtained an $O(1)$-approximation.  In this paper we extend the
ideas to handle \prob{Elem-SNDP} by considering node-minimal feasible
solutions of $0$-$1$-skew-bisupermodular functions (based on bisets)
that arise in the augmentation framework for \prob{Elem-SNDP}
\cite{JainMVW02,FleischerJW06,CheriyanVV06,Nutov12}. An
important and crucial aspect of the algorithm, which also applied to
the results in \cite{ChekuriEV12}, is that our results only apply for
covering a restricted class of skew-bisupermodular functions that
satisfy additional properties. We provide an example later to
illustrate the reason why this is necessary.

As in \cite{DemaineHK14} we use planarity only in one step of the
analysis where we argue about the average degree of a certain graph
that is a minor of the original graph; this is the reason that the
algorithm and analysis generalize to any proper minor-closed family of
graphs. In the interest of clarity and exposition, we have not
attempted to optimize the constants in the approximation analysis; the
bound could perhaps be improved with a more careful analysis.

\medskip\noindent {\bf Other related work:}
Moldenhauer~\cite{Moldenhauer13} showed that an improved analysis of
the algorithm of Demaine \etal \cite{DemaineHK14} reduces their
$6$-approximation guarantee of the node-weighted \prob{Steiner Forest}
in planar graphs to a bound of $3$ which is tight for the algorithm.
Moreover, he claimed, via a different algorithm, a
$9/4$-approximation. However, Berman and Yaroslavtsev~\cite{BermanY12}
showed that the result of~\cite{Moldenhauer13} suffers from a mistake
in the analysis and that the correct approximation guarantee for the
algorithm in \cite{Moldenhauer13} is $18/7$. \cite{BermanY12}
developed an algorithm for the node-weighted planar \prob{Steiner
  Forest} with a $2.4$-approximation. The papers
\cite{Moldenhauer13,BermanY12} point out the connection between
node-weighted \prob{Steiner Forest} and \prob{(Subset) Feedback Vertex
  Set} in planar graphs for which Goemans and
Williamson~\cite{GoemansW98} had developed primal-dual algorithms ---
more details can be found in the papers.

Network design is a broad area that has been explored in depth over
the years. We refer the reader to \cite{GoemansW97} for a survey on
designing algorithms for network design problems using the primal-dual
method, and to recent surveys \cite{KortsarzN10,GuptaK11} for an
overview of the literature and references. We borrow several ideas
from work on \prob{Elem-SNDP}
\cite{JainMVW02,CheriyanVV06,FleischerJW06,Nutov12} and \prob{VC-SNDP}
\cite{ChuzhoyK12,Nutov12}. Some variants of the \prob{SNDP} problem such
as \prob{Prize-collecting SNDP}, \prob{Budgeted SNDP} and
\prob{Network Activation} have been studied in node-weighted setting
and we refer readers
to~\cite{Nutov10,ChekuriEV12-approx,Vakilian13,BateniHL13,Fukunaga17}
for developments in these directions, several of which have happened
after the conference version of this paper appeared.

\bigskip\noindent {\bf Organization:} Section~\ref{sec:prelim} sets up
the relevant technical background on bisets and certain abstract
properties that play a critical role in the
analysis. Section~\ref{sec:node-wt-elem-sndp} outlines the
augmentation framework and the properties of the requirement function
that arises in each phase of the augmentation
framework. Section~\ref{sec:main-result} (in particular,~\ref{sec:main-thm})
describes the main technical result of this paper which is a constant
factor approximation algorithm for the problem of covering
node-weighted biuncrossable functions that arise in the augmentation
framework. In Section~\ref{sec:node-wt-vc-sndp} we apply the framework
to handle $\set{0,1,2}$-\prob{VC-SNDP}. We discuss some open problems
and conclude in Section~\ref{sec:conclusions}.

\section{Preliminaries}
\label{sec:prelim}
We formally define node-weighted \prob{EC-SNDP} and
\prob{Elem-SNDP}. The input to node-weighted \prob{EC-SNDP} is an
undirected graph $G=(V,E)$, a weight function
$w: V \rightarrow \mathbb{R}_+$, and a non-negative integer
requirement $r(uv)$ for each unordered pair of nodes $uv$. The goal is
to find a minimum weight subgraph $H=(V_H,E_H)$ such that $H$ has
$r(uv)$ edge-disjoint paths for each node pair $uv$; by the weight of
$H$ we mean $w(V_H)$ since we are considering node-weights. Call a
node $u$ a terminal if it participates in a pair $uv$ such that
$r(uv) > 0$. Any feasible solution $H$ contains all terminals and
hence we can assume without loss of generality that the weight of
terminals is zero.

The input to node-weighted \prob{Elem-SNDP} is an undirected graph
$G=(V,E)$ along with a partition of $V$ into reliable nodes $R$ and
non-reliable nodes $V \setminus R$. The \emph{elements} of $G$ are
$E \cup (V\setminus R)$. The input also specifies a non-negative
weight function $w: V \rightarrow \mathbb{R}_+$ and integer
requirements $r(uv)$ only over pairs $uv$ where both $u,v$ are
\emph{reliable} nodes. The goal is to find a minimum weight subgraph
$H=(V_H,E_H)$ of $G$ such that for each pair $uv$ of reliable nodes,
$H$ has $r(uv)$ element-disjoint paths. For simplicity we can assume
that $R$ forms an independent set in $G$ by sub-dividing each edge
$uv$ where both $u,v \in R$ and adding a new non-reliable node. Then
element-disjoint paths correspond to paths that are disjoint on
non-reliable nodes. Once again, we can assume that any node $u$ that
participates in a pair $uv$ such that $r(uv) > 0$ can be assumed to
have zero weight since it has to be included in every feasible
solution. It is straightforward to see that \prob{EC-SNDP} is a
special case of \prob{Elem-SNDP} in which all nodes in the input graph
are reliable.

Following the general approach from prior work we reduce \prob{SNDP}
to a more abstract problem of covering certain set and biset
requirement functions. We set up the desired notation and definitions
for this purpose and state several basic properties. We borrow
extensively from past work \cite{Nutov12, CheriyanVV06, FleischerJW06}
and give a few proofs here and some in the
Appendix~\ref{subsec:omitted-preliminaries} for the
sake of completeness.

A key definition is that of \emph{biset}. A biset is a pair of sets
$\hS = (S, S') \in 2^V \times 2^V$ such that $S \subseteq S'$. The set
$S$ is the \emph{inner part} of $\hS$, $S'$ is the \emph{outer part}
of $\hS$, and $S'\setminus S$ is the \emph{boundary} of $\hS$ which is
also denoted by $\bd(\hS)$.  We define the $\subseteq$ relation on the
bisets as follows. $\hS \subseteq \hT$ iff $S \subseteq T$ and
$S' \subseteq T'$.  We use the teriminology $\hS \subset \hT$ if
$\hS \subseteq \hT$ and $\hS \neq \hT$.  We define the following
operations on bisets.  The union, intersection, and difference of
$\hS$ and $\hT$ are defined as
$\hS \cap \hT = (S \cap T, S' \cap T')$,
$\hS \cup \hT = (S \cup T, S' \cup T')$, and
$\hS \setminus \hT = (S \setminus T', S' \setminus T)$.

The following two propositions are straightforward to verify.
\begin{prop}\label{prop:relation-prop}
  The $\subseteq$ relation is a partial order over the bisets.
\end{prop}

\begin{prop}
  $\hS \cap \hT \subseteq \hS$ and $\hS \setminus \hT \subseteq \hS$.
\end{prop}

\begin{definition}[\prob{Crossing Bifamily}]
  A family of bisets $\mP$ is \emph{crossing} iff, for any bisets
  $\hS$ and $\hT$ in $\mP$, union, intersection and differences of
  $\hS$ and $\hT$ are in $\mP$.
\end{definition}

\begin{definition}[\prob{Bimaximal Function}]
  Let $\mP$ be a crossing bifamily.  A function
  $f: \mP \rightarrow \mathbb{Z}$ is bimaximal iff, for any
  $\hS, \hT \in \mP$ whose inner parts are disjoint (that is, $S \cap T = \emptyset$),
  \begin{equation}
    f(\hS \cup \hT) \leq \max\{f(\hS), f(\hT)\}
    \label{eqn:bimaximal}
  \end{equation}
\end{definition}

\begin{definition}[\prob{Bisubmodular Function}]
  Let $\mP$ be a crossing bifamily.
  A function $f: \mP \rightarrow \mathbb{Z}$ is
  bisubmodular iff for any $\hS, \hT\in\mP$, both of the
  following inequalities hold\footnote{Note that these
    inequalities hold for symmetric submodular set functions and
    we work with symmetric biset functions throughout the paper
    since the graphs are undirected.}:
  \begin{align}
    f(\hS) + f(\hT) &\geq f(\hS \cap \hT) +
                      f(\hS \cup \hT) \label{eqn:bisubmod-one}\\
    f(\hS) + f(\hT) &\geq f(\hS \setminus \hT) +
                      f(\hT \setminus \hS) \label{eqn:bisubmod-two}
  \end{align}
  A function $f$ is \emph{bisupermodular} iff $-f$ is
  bisubmodular.
\end{definition}

\begin{definition}[\prob{Skew-bisupermodular Function}]
  Let $\mP$ be a crossing bifamily.
  A function $f: \mP \rightarrow \mathbb{Z}$ is
  skew-bisupermodular iff for any $\hS, \hT \in \mP$, one of the following holds:
  \begin{eqnarray}
    f(\hS \cap \hT) + f(\hS \cup \hT) &\geq f(\hS) +
                                        f(\hT) \label{eqn:skew-bisuper-one}\\
    f(\hS \setminus \hT) +
    f(\hT \setminus \hS) &\geq f(\hS) + f(\hT) \label{eqn:skew-bisuper-two}
  \end{eqnarray}
\end{definition}

\begin{definition}[\prob{Biuncrossable Function}]
  Let $\mP$ be a crossing bifamily.
  A function $f: \mP \rightarrow \mathbb{Z}$ is
  biuncrossable iff for any $\hS,\hT \in \mP$ such that $f(\hS) > 0$ and $f(\hT) > 0$,
  one of the following holds:
  \begin{eqnarray}
    f(\hS \cap
    \hT) + f(\hS \cup \hT) &\geq& f(\hS) +
                                  f(\hT) \label{eqn:biuncross-one}\\
    f(\hS \setminus \hT) +
    f(\hT \setminus \hS) &\geq& f(\hS) + f(\hT)
                                \label{eqn:biuncross-two}
  \end{eqnarray}
\end{definition}

\begin{prop}[Lemma 3.8 in~\cite{FleischerJW06}]\label{prop:difference-bisup}
  Let $f$ be a skew-bisupermodular function and let $g$ be a
  bisubmodular function on the same domain. Then $f - g$ is a
  skew-bisupermodular function.
\end{prop}


\begin{prop}\label{prop:boundary-bisubmodular}
  The $\card{\bd(.)}$ function is a bisubmodular function over the set
  of all bisets over $V$.
\end{prop}

Let $G=(V,E)$ be a graph and let $\hS$ be a biset over $V$.  We say
an edge $e$ crosses $\hS$ if $e$ has one endpoint $S$ and the other
endpoint in $V \setminus S'$.  For any $F\subseteq E$ we let
$\delta_F(\hS)$ denote the set of all edges in $F$ that cross $\hS$.
We define $\Gamma_F(\hS)$ to be the set of all vertices
$u \in V \setminus S'$ for which there exists an edge $uv \in F$ that
crosses $\hS$. For a subgraph $H$ we abuse notation and use
$\delta_H(\hS)$ to denote $\delta_{E(H)}(\hS)$ and $\Gamma_H(\hS)$ to
denote $\Gamma_{E(H)}(\hS)$.

\begin{lemma}[Lemma 3.7 in~\cite{FleischerJW06}]
  \label{lem:delta-bisubmod}
  For any graph $G=(V, E)$ and any subset of edges $F\subseteq E$,
  $\card{\delta_F(.)}$ is bisubmodular over $\mP$ where $\mP$ is
  any crossing family of bisets over $V$.
\end{lemma}

\begin{definition}[\prob{Feasible Cover}]
  Let $G=(V,E)$ be a graph and let $f$ be an integer-valued function
  defined on a collection of bisets $\mP$ over $V$. We say that
  $F \subseteq E$ is a feasible cover of $f$ if
  $\card{\delta_{F}(\hS)} \geq f(\hS)$ for each $\hS \in \mP$.  We say
  that a subgraph $H=(V(H),E(H))$ is a feasible cover of $f$ if $E(H)$
  is a feasible cover of $f$.  A subgraph $H$ is a node-minimal cover
  of $f$ if $H \setminus \set{v}$ is not a feasible cover of $f$ for any
  $v \in V(H)$.
\end{definition}


\begin{definition}[\prob{(Minimal) Violated Biset}]
 Let $G=(V,E)$ be a graph and let $f$ be an integer-valued function
 defined on a collection of bisets $\mP$ over $V$.
 For $F \subseteq E$ we say that a biset $\hS \in \mP$ is violated with
 respect to $F$ if $\card{\delta_{F}(\hS)} < f(\hS)$. We say that
 $\hS$ is a minimal violated biset with respect to $F$ if $\hS$ is
 violated and there is no violated biset $\hT$ such that $\hT \subset \hS$.
 These definitions extend to violation with respect to a subgraph $H$
 of $G$.
\end{definition}

\begin{definition}[\prob{Non-overlapping Bisets}]
  Two bisets $\hS$ and $\hT$ are \emph{non-overlapping} iff
  one of the following holds:
  \begin{enumerate}[(i)]
    \vspace{-.1in}
  \item $\hS \subseteq \hT$ or $\hT \subseteq \hS$.
    \vspace{-.1in}
  \item The sets $S' \cap T$ and $S \cap T'$ are empty.
  \end{enumerate}
  If the bisets do not satisfy any of the above conditions, they
  are \emph{overlapping}.
\end{definition}

A useful observation that we will need later is that minimal violated
bisets do not overlap with other (not necessarily minimal)
violated bisets.

\begin{lemma} \label{lem:minimal-violated-crossing}
  Let $h$ be a $\set{0, 1}$-biuncrossable function. Let $\hC$ be a minimal violated
  biset of $h$ and let $\hS$ be a violated biset of
  $h$. Then, $\hC$ and $\hS$ do not overlap. In particular,
  the inner parts of the minimal violated bisets of
  $h$ are disjoint.
\end{lemma}
\begin{proof}
  Since $h(\hC) = h(\hS) = 1$ and $h$ is a biuncrossable function,
  $h(\hC \cap \hS) = h(\hS \cup \hS) = 1$ or $h(\hC \setminus \hS)
  = h(\hS \setminus \hC) = 1$.
  Suppose that the former case holds. Since $\hC
  \cap \hS \subseteq \hC$, it follows from the minimality of $\hC$
  that $\hC \cap \hS = \hC$. Thus $\hC \subseteq \hS$ and hence $\hC$
  and $\hS$ are non-overlapping.
  Therefore we may assume that $h(\hC \setminus \hS) = h(\hS \setminus \hC) = 1$.
  Since $\hC \setminus \hS \subseteq \hC$, it follows from the minimality of
  $\hC$ that $\hC \setminus \hS = \hC$. Thus the sets $C \cap S'$ and $C'
  \cap S$ are empty, and hence $\hC$ and $\hS$ are non-overlapping.

  If $\hC_1$ and $\hC_2$ are both minimal violated bisets of $h$, none of them is a subset
  of the other one; hence, $C_1 \cap C'_2 =\emptyset$ which implies that $C_1 \cap C_2$
  is empty as well.
 \end{proof}

Since we are interested in node-weighted problems the subgraphs that
arise in our algorithms and analysis are typically node-induced
subgraphs. We use the standard terminology of $G[S]$ to denote the
subgraph of $G$ induced by a node subset $S \subseteq V(G)$.
We use $E[S]$ to denote the set of edges with both end points in
$S$. The graph $G[S] = (S, E[S])$ is the subgraph induced by the
by the vertex set $S$. We frequently need to consider the graph
$(V, E[S])$ and when there is no confusion we use $G[S]$ to denote
this graph as well.

\section{Algorithm for Node-weighted Elem-SNDP}
\label{sec:node-wt-elem-sndp}
In this section, we formally set up the augmentation framework for
node-weighted \prob{Elem-SNDP}. We point out the specific features of
the optimization problem that arises in each phase of the augmentation
framework. Our main technical result which provides an
$O(1)$-approximation for the augmentation problem of each phase is
formally described and analyzed in Section \ref{sec:main-result}.

\subsection{\prob{Elem-SNDP} and Covering Skew-Bisupermodular Functions}
We set up \prob{Elem-SNDP} as a special case of covering skew-bisupermodular
functions using Menger's theorem for element connectivity.
Given an instance of \prob{Elem-SNDP} over a graph $G=(V,E)$ with requirements
specified by $r$, we extend the requirements to bisets as follows.
For each biset $\hS$ defined on $V$, $\rE(\hS)$ is defined as
$\max_{u \in S, v \in V \setminus S'} r(uv)$; in other words,
$\rE(\hS)$ is the maximum connectivity requirement over all pair of
vertices that are \emph{separated} by $\hS$. Note that we only have
connectivity requirement over pairs of reliable nodes.
Let $\mPE$ be the
collection of all bisets defined on $V$ whose boundaries only contain
\emph{non-reliable} nodes. Then we define
$\fE:\mPE \rightarrow \mathbb{Z}_{+}$ as
$\fE(\hS) = \rE(\hS) - \card{\bd(\hS)}$. The following theorem
is not hard to prove and can be found in~\cite{FleischerJW06,CheriyanVV06}.

\begin{theorem}[\bf Menger's theorem for element connectivity]
  \label{thm:Elem-menger}
  Let $G=(V,E)$ be an undirected graph with $V$ partitioned into
  reliable nodes $R$ and non-reliable nodes $V \setminus R$. Two
  distinct nodes $s,t \in R$ are $k$-element connected iff for each biset $\hS\in \mPE$
  separating $s$ and $t$, $\card{\delta(\hS)} +  \card{\bd(\hS)}\geq k$.
\end{theorem}

Applying Menger's theorem, solving node-weighted \prob{Elem-SNDP} is
equivalent to finding a minimum node-weighted (feasible) cover of $\fE$.

\begin{prop}\label{prop:collection-mPE}
  $\mPE$ is a crossing bifamily.
\end{prop}
\begin{proof}
  $\hS \cap \hT \subseteq \hS$ and $\hS \setminus \hT \subseteq \hS$
  which implies that $\bd(\hS \cap \hT) \subseteq \bd(\hS)$ and
  $\bd(\hS \setminus \hT) \subseteq \bd(\hS)$; hence intersection and
  difference preserve the property that the boundary does not contain
  any reliable nodes. Recall that
  $\hS \cup \hT = (S \cup T, S' \cup T')$ and hence
  $\bd(\hS\cup\hT) = (S' \cup T') \setminus (S \cup T) \subseteq
  \bd(\hS) \cup \bd(\hT)$; therefore, if both $\hS,\hT\in \mPE$, then
  $\hS\cup\hT \in \mPE$ as well.
\end{proof}


\begin{prop} \label{prop:elem-sndp-requirement} Let $\rE$ be the
  requirement function arising from an instance of
  \prob{Elem-SNDP} (in other words, $\rE$ is defined on the crossing bifamily $\mPE$). Then,
  \begin{itemize}
		\vspace{-.1in}
			\item $\rE(\hS) = 0$ for all bisets $\hS$
    				such that $S = \emptyset$ or $S' = V$.
		\vspace{-.1in}
			\item $\rE$ is skew-bisupermodular and
    		bimaximal.
  \end{itemize}
\end{prop}
\begin{proof}
It is straightforward to see that if $S=\emptyset$ or $S'=V$, then $\hS$
does not separate any pair of terminals and $\rE(\hS)=0$.

Fleischer \etal proved that $\rE$ is skew-bisupermodular on $\mPE$
(see Lemma 3.11 in~\cite{FleischerJW06}). Further, we show that $\rE$
is bimaximal on $\mPE$. Let $\hS, \hT \in \mPE$ and let $(s, t)$ be a
pair of terminals that have the maximum connectivity requirement among
all terminal pairs separated by $\hS \cup \hT$, i.e.,
$\rE(\hS \cup \hT) = r(s, t)$. Since $s \in S \cup T$, we have
$s \in S$ or $s \in T$; without loss of generality, assume $s \in
S$. Since $t \in V \setminus (S' \cup T')$, the pair $(s, t)$ is
separated by $\hS$ and thus
$\rE(\hS \cup \hT) \leq \rE(\hS) \leq \max\set{\rE(\hS), \rE(\hT)}$.
\end{proof}

\subsection{Augmentation Framework}\label{sec:augmentation}

Now we turn our attention to the proof of
Theorem~\ref{thm:intro-elem-sndp}.  
We alert the reader that, in order to cover the function $\fE$, we need
to pick a set of \emph{edges}. But since the weights are (only) on
the nodes, we pay for a set of nodes and we can use any of the edges
in the graph induced by the selected nodes to cover $\fE$.
More precisely, our goal is to select a minimum-weight subgraph
$H=G[X]$ that covers $\fE$, where $X$ is a subset of the vertex set of $G$. We will
always assume that $X$ contains all terminals.

Our algorithm for covering $\fE$ uses the augmentation framework
introduced by Williamson \etal \cite{WilliamsonGMV95} for
edge-weighted \prob{EC-SNDP}. For a non-negative integer $\ell$
consider the requirement function $r_\ell$ where
$r_\ell(\hS) = \min \set{\ell, \rE(\hS)}$. Similarly we define
$f_\ell$ where $f_\ell(\hS) = r_\ell(\hS) - \card{\bd(\hS)}$.  The
algorithm performs $k$ phases with the following property: at the end
of phase $\ell$, the algorithm constructs a subgraph $H_\ell$ that
covers $f_\ell$.  In phase $\ell$, the algorithm starts with the
subgraph $H_{\ell - 1}= (V, E[X_{\ell-1}])$ that covers $f_{\ell - 1}$ and
adds a new set of vertices to $H_{\ell - 1}$ to obtain
$H_\ell = (V,E[X_\ell])$. We elaborate on this augmentation process.  It
is convenient to assume that all vertices in $X_{\ell-1}$ have zero
weight since they have already been paid for.  Let
$G_\ell = (V, E(G) \setminus E(H_{\ell - 1}))$.  The goal in phase
$\ell$ is to select a minimum-weight subgraph $H$ of $G_\ell$ that
covers the function $h_\ell$, where
$h_\ell(\hS) = \max\{0,f_\ell(\hS) - \card{\delta_{H_{\ell - 1}}(\hS)}\}$ for
each $\hS \in \mPE$. Note that $h_\ell(\hS) \le 1$ for all $\hS$.
Moreover it is an uncrossable and bimaximal function, and satisfies
certain other properties which we will formally specify later.

The phase $\ell$ augmentation problem is then the following: given a
subgraph $H_{\ell-1}$ that covers $f_{\ell-1}$, find a
minimum weight subset of nodes $A$ such that $(V,E_{G_\ell}[X_{\ell-1} \cup A])$
covers $h_\ell$.

\begin{theorem}\label{thm:aug-approx}
  Suppose there is a $\lambda(\ell)$-approximation algorithm for the
  phase $\ell$ augmentation problem for each  $1 \le \ell \le k$.
  Then there is a $\sum_{\ell=1}^k \lambda(\ell)$-approximation for
  node-weighted \prob{Elem-SNDP}.
\end{theorem}

The preceding theorem is an easy consequence of the augmentation
framework and the fact that the optimum cost of any instance of the
augmentation problem that arises in phase $\ell$ is upper bounded by
the optimum cost of the solution for the original instance of
\prob{Elem-SNDP}. We will show that $\lambda(\ell) \le 10$ if $G$ is
planar. For proper minor-closed family of graphs we prove that
$\lambda(\ell) = O(1)$ where the constant depends on the family.
This leads to the claimed $O(k)$ approximation for node-weighted
\prob{Elem-SNDP} that proves Theorem~\ref{thm:intro-elem-sndp}.

\subsection{Properties of the Function $h_\ell$}
We now discuss some properties of the function that arises in the
augmentation process.

\begin{lemma}\label{lem:phase-supermodular}
  The functions $r_{\ell}$, $f_\ell$ and
  $f_\ell - \card{\delta_{H_{\ell - 1}}}$ are
  skew-bisupermodular on $\mPE$.
\end{lemma}
\begin{proof}
  Note that $r_{\ell}$ can be defined as the biset function
  corresponding to the \prob{Elem-SNDP} instance in which $r_{\ell}(s,
  t) = \min\{\ell, r(s,t)\}$. By Proposition~\ref{prop:elem-sndp-requirement},
  $r_{\ell}$ is a skew-bisupermodular function.
  Since $\card{\bd(\hS)}$ is a bisubmodular function
  (Proposition~\ref{prop:boundary-bisubmodular}), by
  Proposition~\ref{prop:difference-bisup}, $f_{\ell}$ is skew-bisupermodular as well.
  Moreover, using the fact that $f_{\ell}$ is skew-bisupermodular and
  $\card{\delta_{H_{\ell-1}}(.)}$ is bisubmodular, by
  Proposition~\ref{prop:difference-bisup}, $f_\ell - \card{\delta_{H_{\ell - 1}}}$ is
  skew-bisupermodular as well.
\end{proof}

For each biset $\hS\in \mPE$, Let $h'_\ell(\hS) = f_\ell(\hS) - \card{\delta_{H_{\ell-1}}(\hS)}$. From the
preceding lemma $h'_\ell$ is skew-bisupermodular and moreover for each $\hS \in \mPE$,
$h'_\ell(\hS) \le 1$. Note that $h_\ell(\hS) = \max\{0, h'_\ell(\hS)\}$.
We claim that $h_\ell$ is bi-uncrossable.
To see this, suppose $h_\ell(\hS) = 1$ and $h_\ell(\hT) = 1$
then by skew-supermodularity of $h'_\ell$ we have
$h'_\ell(\hS \cup \hT) + h'_\ell(\hS \cap \hT) \ge 2$ or
$h'_\ell(\hS \setminus \hT) + h'_\ell(\hT \setminus \hS) \ge 2$;
This is possible only if $h_\ell(\hS \cup \hT) + h_\ell(\hS \cap \hT) \ge 2$ or
$h_\ell(\hS \setminus \hT) + h_\ell(\hT \setminus \hS) \ge 2$ because
both $h_\ell$ and $h'_\ell$ are at most $1$ on any biset.

\begin{prop} \label{prop:violated-pairs-alternate-defn}
	Consider an integer $\ell\leq k$. Then $h_\ell(\hS) = 1$ iff $\rE(\hS) \geq \ell$ and $\card{\bd (\hS)} +
	\card{\delta_{H_{\ell - 1}}(\hS)} = \ell - 1$.
\end{prop}
\begin{proof}
  If $\rE(\hS) \geq \ell$ and
  $\card{\bd(\hS)} + \card{\delta_{H_{\ell - 1}}(\hS)} = \ell - 1$,
  then by definition, $h_{\ell}(\hS) = \max\{0, \ell  - \bd(\hS) -
  \card{\delta_{H_{\ell - 1}}(\hS)}\} = 1$.

  We now consider the other direction. Suppose that $h_\ell(\hS) =
  1$.  This implies that
  $f_\ell(\hS) - \card{\delta_{H_{\ell - 1}}(\hS)} = 1$. Since
  $H_{\ell - 1}$ covers $f_{\ell - 1}$, we have
  $\card{\delta_{H_{\ell - 1}}(\hS)} \geq f_{\ell - 1}(\hS)$. Thus
  $f_\ell(\hS) \geq f_{\ell - 1}(\hS) + 1$ and hence
  $r_\ell(\hS) \geq r_{\ell - 1}(\hS) + 1$.  It follows that
  $\rE(\hS) \geq \ell$ and
  $\card{\bd (\hS)} + \card{\delta_{H_{\ell - 1}}(\hS)} \leq
  r_\ell(\hS) - 1 = \ell - 1$.  Moreover, since
  $\card{\delta_{H_{\ell - 1}}(S)} + \card{\bd(\hS)} \geq r_{\ell -
    1}(\hS) = \ell - 1$,
  $\card{\delta_{H_{\ell - 1}}(\hS)} + \card{\bd(\hS)} = \ell - 1$.
\end{proof}


Recall that $H_{\ell-1}$ covers $f_{\ell-1}$ and $G_\ell = G(V, E(G)\setminus
E(H_{\ell-1}))$. Further $H_{\ell-1} =(V, E[X_{\ell-1}])$ where $X_{\ell-1}$
is the set of nodes paid for in the first $\ell-1$ phases.

\begin{lemma}
  \label{lem:aug-func-bd}
  For any $X \supset X_{i-1}$ let $H_X = (V, E_{G_{\ell}}[X])$ be a subgraph
  of $G_\ell$. Suppose $\hC \in \mPE$ is a violated biset
  of $H_X$ with respect to $h_\ell$. Then $\bd(\hC) \subseteq X$.
\end{lemma}
\begin{proof}
  Suppose for the sake of contradiction there is a violated biset
  $\hC$ and a vertex $u \in \bd(\hC)$ such that $u \not \in X$.
  Consider the biset $\hC_1 = (C, \bd(\hC) \setminus \{u\})$.  Since
  $u$ is not a terminal, $r_\ell(\hC_1) = r_\ell(\hC)  = \ell$.  Since
  $h_\ell(\hC) = 1$ we have $r_\ell(\hC) = \ell$ and
  $\card{\bd(\hC)} + \card{\delta_{H_{\ell-1}}(\hC)} = \ell - 1$.  In
  the graph $H_{\ell-1}$, the vertex $u$ has no edges incident to it
  since $H_{\ell-1} = (V, E[X_{\ell-1}])$ and $u \not \in X_{\ell-1}$.
  Therefore, $\delta_{H_{\ell-1}}(\hC_1) = \delta_{H_{\ell-1}}(\hC)$.
  Since $\card{\bd(\hC_1)} = \card{\bd(\hC)} - 1$, we have
  $\card{\bd(\hC_1)} + \card{\delta_{H_{\ell-1}}(\hC_1)} = \ell -2$
  which implies that $H_{\ell-1}$ is not a feasible cover for
  $r_{\ell-1}$, a contradiction.
\end{proof}

\begin{lemma}
  \label{lem:aug-func-minimal}
  For any $X \supset X_{i-1}$ let $H_X = (V, E_{G_\ell}[X])$ be a subgraph
  of $G_\ell$. Suppose $\hC \in \mPE$ is a \emph{minimal} violated biset
  of $H_X$ with respect to $h_\ell$. Then the following properties
  hold.
  \begin{itemize}
  \item $C' \subseteq X$.
  \item $G[C]$ is a connected subgraph of $G$.
  \end{itemize}
\end{lemma}
\begin{proof}
  For ease of notation we let $H$ denote the subgraph $H_X$. Note that
  Since $\hC$ is a violated biset in $H$ we have $h_\ell(\hC) = 1$
  and $|\delta_{H}(\hC)| = 0$. From
  Proposition~\ref{prop:violated-pairs-alternate-defn}, since
  $h_\ell(\hC) = 1$, $r(\hC) = \ell$ and
  $\card{\bd(\hC)} +\card{\delta_{H_{\ell - 1}}(\hC)} = \ell -
  1$. Suppose there is a vertex $u \in C'$ such that $u \not \in
  X$. By Lemma~\ref{lem:aug-func-bd}, $u\in C$. First, $u$ is not a terminal since all terminals are in
  $X_{i-1}$ (and hence in $X$).  Second $u$ is an isolated vertex in
  $H$ since the only edges in $H$ are between nodes in $X$. Consider
  the biset $\hC_1 = (C -u, C'-u)$ obtained from $C$ by removing $u$.
  Since $u$ is not a terminal we have $r(\hC) = r(\hC_1)$.  And since
  $u$ is isolated in $H$ we have
  $\delta_{H}(\hC_1) = \delta_{H}(\hC)$, since we moved $u$ out of
  $C'$, $\bd(\hC_1) \subseteq \bd(\hC)$.  These facts imply that
  $h_\ell(\hC_1) = 1$ and $\hC_1$ is a violated biset in $H$. This
  contradicts minimality of $\hC$. Therefore $C' \subseteq X$.

  We now prove that $G[C]$ is connected. For sake of contradiction
  suppose it is not. Let $C_1, C_2$ be two non-empty sets that
  partition $C$ such that there is no edge between $C_1$ and $C_2$ in
  $G$; such a partition exists if $G[C]$ is not connected.  Note that
  $E_{H}(C_1,C_2) = \emptyset$ since $H$ is a subgraph of $G$.
  Define $\hC_1=(C_1, C_1 \cup \bd(\hC))$ and
  $\hC_2=(C_2, C_2 \cup \bd(\hC))$.  Since $r$ is bimaximal
  (Proposition~\ref{prop:violated-pairs-alternate-defn}),
  $r(\hC) \leq \max\set{r(\hC_1), r(\hC_2)}$.  Thus, without loss of
  generality we can assume that $r(\hC_1) \geq r(\hC) \geq \ell$.
  Since $E_{H}(C_1,C_2) = \emptyset$ we have $\delta_{H}(\hC_1)
  \subseteq \delta_{H}(\hC)$. Since $\bd(\hC_1) = \bd(\hC)$ and
  $\hC$ was a violated biset it follows that $\hC_1$ is also
  a violated biset with respect to $h_\ell$ in $H$. This contradicts
  the minimality of $\hC$.
\end{proof}

The collection of minimal violated bisets in $H_X$ with respect to
$h_\ell$ are disjoint (due to the biuncrossability), and they can computed
in polynomial time via Menger's theorem for element-connectivity
and standard maxflow algorithms. We refer the reader to
\cite{JainMVW02,FleischerJW06}.

\section{Approximation Algorithm for the Augmentation Problem}
\label{sec:main-result}
In this section we design an $O(1)$-approximation algorithm for
minimum node-weighted cover for the augmentation problem that needs to
be solved in each of the $k$ phases of the augmentation framework that
was described in the preceding section. We recall the problem arises in the 
phase $\ell$ of the augmentation framework. 
We are given $X_{\ell-1} \subseteq V$ such that
the graph $H_{\ell-1} = (V, E[X_{\ell-1}])$ is a feasible cover for
$f_{\ell-1}$. The goal is to find a minimum-weight subset of nodes
$A \subseteq V \setminus X_{\ell-1}$ such that
$H_\ell = (V, E[X_{\ell-1}\cup A])$ covers $f_\ell$. This is recast as
the problem of covering the $\{0,1\}$-biuncrossable function $h_\ell$
in the graph $G_\ell = (V, E(G) \setminus E(H_{\ell-1}))$.

In the edge-weighted case one can obtain a $2$-approximation for
covering a $\{0,1\}$-biuncrossable function in a general undirected
graph provided the function has some reasonable computational
properties such as the ability to efficiently find the minimal
violated sets with respect to any subset of the given edges of a
graph. In particular a natural LP relaxation has an integrality gap of
at most $2$. However, in the node-weighted setting, a natural LP
relaxation that we will discuss shortly has an unbounded integrality
gap. However, the function $h_\ell$ that arises in the augmentation
framework for \prob{Elem-SNDP} has additional properties that allow us
to prove a constant factor approximation in planar graphs. One can
also prove an $O(\log n)$-approximation in general graphs and this
cannot be improved since node-weighted Steiner tree is a special case
which generalizes the Set Cover problem.

In this section we prove the following theorem.
\begin{theorem}\label{thm:main}
  Let $G=(V, E)$ be a node-weighted graph from a proper minor-closed
  family of graphs $\mG$. Let $h_\ell$ be $\{0,1\}$-biuncrossable
  function that arises in phase $\ell$ of the augmentation framework
  for an instance of \prob{Elem-SNDP} defined over the graph $G$.
  There exists an $O(1)$-approximation algorithm for the problem of
  finding a minimum-weight node subset to cover $h_\ell$.
\end{theorem}

We remark that our result applies to a class of $\{0,1\}$-biuncrossable 
functions that is more general than the class of functions that arise from
\prob{Elem-SNDP}. Characterizing the precise class for which
the algorithm applies is not quite as clean as
we would like and hence we do not attempt to do so.

Our algorithm is a primal-dual algorithm modeled after the
well-studied algorithm for the edge-weighted case
\cite{AgrawalKR95,GoemansW95}.  The adaptation of the primal-dual
algorithm to the node-weighted case in planar graphs was done in
\cite{DemaineHK14} but were concerned with Steiner forest and
$\{0,1\}$ proper functions while our setting is more general.

\subsection{A Primal-Dual Algorithm}\label{subsec:primal-dual}
Instead of focusing on the specific setting of covering the restricted
class of functions that arise in \prob{Elem-SNDP} we will work in a
abstract framework where we have a general $\{0,1\}$-biuncrossable
function $h$ defined over a crossing bifamily $\mP$ in a node-weighted
graph $G=(V,E)$. The goal is to find a minimum weight subset $X
\subseteq V$ such that the subgraph $H = (V,E[X])$ covers $h$, that
is, $|\delta_H(\hS)| \ge h(\hS)$ for each $\hS \in \mP$.

\medskip\noindent
{\bf LP relaxation.} We consider a natural LP-relaxation of the problem and its dual which
are shown in Figure~\ref{fig:undirected-lp}. There is a variable
$x(v)$ which in the integer programming formulation indicates whether
$v$ is chosen and in the LP relaxation $x(v)$ is relaxed to be in the
interval $[0,1]$. Consider a biset $\hS \in \mP$ such that
$h(\hS) = 1$. Then any subgraph $H$ of $G$ that covers $h$ needs to
contain an edge $e \in \delta_G(\hS)$ which implies that there is
an endpoint $v$ of $e$ such that $v \in \Gamma(\hS)$; therefore at
least one vertex in $\Gamma(\hS)$ needs to be included in any feasible
cover of $h$.  This justifies the constraint in the LP relaxation.
Note that we omitted the constraint $x(v) \leq 1$ from the primal
since it is redundant. 

\begin{lemma}\label{lem:valid-lp}
  The \prob{Primal-LP} is a valid relaxation of the problem of
  covering $0$-$1$ biset functions on node-weighted graphs.
\end{lemma}

\begin{figure}[!h]
\begin{center}
\begin{boxedminipage}{0.38\textwidth}
\underline{\textbf{\prob{Primal-LP}}
\Comment{Input: $(G, \mP, h)$}}
\vspace{-.15in}
\begin{align*}
	\min & \sum_{v \in V} w(v) x(v) \\
	\text{s.t.} & \sum_{v \in \Gamma(\hS)} x(v) \geq h(\hS) \quad
	\forall \hS \in \mP\\
	& x(v)  \geq 0 \quad \forall v \in V
\end{align*}
\end{boxedminipage}
\hspace{0.4in}
\begin{boxedminipage}{0.40\textwidth}
\underline{\textbf{\prob{Dual-LP}}
\Comment{Input: $(G, \mP, h)$}}
\vspace{-.15in}
\begin{align*}
	\max & \sum_{\hS \in \mP} h(\hS) y(\hS) \\
	\text{s.t.} & \sum_{\hS: v \in \Gamma(\hS)} y(\hS) \leq w(v) \quad
	\forall v \in V\\
	& y(\hS) \geq 0 \quad \forall \hS \in \mP
\end{align*}
\end{boxedminipage}
\end{center}
\vspace{-0.5 cm}
\caption{LP-relaxation of the optimization problem of covering biset functions on node-weighted graphs and its dual program.}
\label{fig:undirected-lp}
\end{figure}
\smallskip\noindent

\medskip\noindent
{\bf Integrality gap example.}
Before we describe the primal-dual algorithm we describe a simple
example to demonstrate that the integrality gap of the LP is unbounded
for general biuncrossable functions. Let $G=(V,E)$ be a complete graph
on $n \ge 3$ nodes. Consider the case when $\mP$ is the set of all
bisets over $V$ and $h$ is the function such that $h(\hS) =1$ for a
biset $\hS$ where $S = \{v_1\}$ and $S' = \{v_1\}$; $h(\hT) = 0$ for
all other bisets. It is easy to see that $h$ is biuncrossable.  Let
$w(v_1) = 1$ and $w(v_i) = 0$ for all $i \ge 2$.  The only way to
cover $h$ is to pick $v_1$ and one other node and hence the optimum
solution has weight $1$. However, the LP relaxation is forced to only
pick a neighbor of $v_1$ and pays $0$.  This is true even if the
variables are required to be integer and in fact the integer solution
is not necessarily even feasible for the original problem. The
technical issue here is that there is no notion of ``terminals'' when
working with a general biuncrossable function. However, when working
with \prob{Elem-SNDP} the terminals are always included in a solution
and can be assumed to have weight $0$. As the algorithm proceeds the
newly added vertices can be treated as terminals and the connectivity
properties satisfied by $h_\ell$ help in this regard. This will become
clearer in the analysis.

\begin{figure}[h]
	\begin{center}
		\begin{algorithmEnv}
			\underline{{\algCover}%
				$\pth{G, h}$}: $\Bigl.$ \qquad
			\+\\ %
			\LineComment{$\textbf{y}$ denotes the variables of \prob{Dual-LP} and initialized to zero}\\
			${\bf y} \leftarrow 0$ \\
			$P_0 \leftarrow \set{v\sep w(v) = 0}$, $P
                        \leftarrow P_0$ \\
                        $\iter \leftarrow 1$ \\
                        \LineComment{$\algViolated{}(G,h,P)$ returns
                          all minimal violated bisets of $h$ {\sf
                            w.r.t} $G[P]$} \\
                        $\mC_i \leftarrow \algViolated{}(G,h,P_{\iter-1})$ \\
			\While $\mC_i \neq \emptyset$ \\
				\> {\bf Increase} $y(\hS)$ uniformly
                                for all $\hS \in \mC_i$ until a
                                constraint of \\
				\> \prob{Dual-LP} becomes tight (for $v, \sum_{\hS: v \in \Gamma(\hS)} y(\hS) = w(v)$) \\
				\> $P_{\iter} \leftarrow P_{\iter-1}
                                \cup \set{v}$, $P \leftarrow P_i$ \\
                                \> $\iter \leftarrow \iter + 1$ \\
                                \> $\mC_i \leftarrow
                                \algViolated{}(G,h,P_{\iter-1})$ \\
                                \\
                             			\LineComment{\emph{reverse-delete} step:}\\
			$Q \leftarrow P$ \\
			\Foreach $v\in Q$ in the reverse of the order in which the \While loop added vertices \\
				\>\If $\algViolated{}(G,h, Q\setminus\set{v}) = \emptyset$ \\%
				\>\> $Q \leftarrow Q \setminus\set{v}$ 
		\end{algorithmEnv}%
	\end{center}
	\vspace{-0.5 cm}
	\caption{A primal-dual algorithm for covering restricted $\{0,1\}$-biuncrossable functions.}
	\label{fig:primal-dual-alg}
\end{figure}

\medskip\noindent
{\bf Primal-dual algorithm.} We now describe the primal-dual algorithm. It is inspired by the one
in \cite{DemaineHK14} which is an adaptation to the node-weighted
setting of the standard primal-dual algorithm for the edge-weighted
case \cite{AgrawalKR95,GoemansW95}.  The algorithm selects a subset of
vertices $P$ such that the graph $(V, E[P])$ covers $h$. However, we
need to include, at the start of the algorithm, a prespecified subset
of vertices whose weight will not be counted in analyzing the
performance of the algorithm; alternatively we can assume that these nodes
have weight $0$. In \cite{DemaineHK14} the prespecified subset is
the set of terminals. In the augmentation framework this is the set
of all the vertices that have been selected in the previous phases
(for phase $1$ this is the set of terminals).

The algorithm has two high-level stages. In the first stage it starts
with $P_0$, the set of all zero weight vertices, and iteratively adds
vertices to $P_0$ as long as there are violated bisets.  This state is
guided by maintaining a dual feasible solution $\textbf{y}$ that is
implicitly initialized to zero. The iterations proceed as follows.
Consider iteration $i$ and let $P_{i-1}$ be the set of all nodes
selected in the first $i - 1$ iterations.
Let $\mC_i$ be the collection of all minimal
violated bisets of $h$ with respect to graph $G_i = (V,E[P_{i-1}])$.
$\mC_i$ can be computed in polynomial-time, and any two bisets in this
family do not cross by the biuncrossability property of $h$.
The algorithm assumes access to a procedure \algViolated{} that outputs
the minimal violated bisets with respect to a given subgraph.
The first stage of the algorithm stops when
$\mC_i$ is empty. Otherwise, in iteration $i$ it increases the dual variables
$\{y(\hC)\}_{\hC \in \mC_i}$ uniformly until a dual constraint for a
vertex $v$ becomes tight, that is, we have $\sum_{\hS: v \in
  \Gamma(\hS)} y(\hS) = w(v)$. If the dual constraint
corresponding to $v$ becomes tight, we add $v$ to $P_{i-1}$ to obtain
$P_{i}$ and move to iteration $i+1$. If several
vertices become tight at the same time, we pick one of them arbitrarily.

The second stage of the algorithm is a \emph{reverse-delete} step.  Let
$P$ be the set of vertices selected by the primal-dual algorithm.  We
select a subset $Q$ of $P$ as follows. We start with $Q = P$. We order
the vertices of $Q$ in the reverse of the order in which they were
selected by the primal-dual algorithm. Let $v$ be the current vertex.
If $(V, E[Q \setminus \set{v}])$ is still a feasible cover for $h$, we
remove $v$ from $Q$. The algorithm outputs the vertices that remain
in $Q$.

A formal description of the described primal-dual algorithm is given
in Figure~\ref{fig:primal-dual-alg}. As we discussed earlier, the
algorithm is not guaranteed to output a feasible solution for an
arbitrary $\{0,1\}$-biuncrossable function. However, we argue below that
it returns a feasible solution for the functions that arise in
the augmentation framework. We will assume that the algorithm is
run on graph $G_\ell$ with function $h_\ell$ and $P_0 = X_{\ell-1}$.

\begin{proposition} \label{prop:no-neighbor}
  At the start of iteration $i$ of the first while loop, we have
  $P_{i-1} \cap \Gamma(\hC) = \emptyset$ for every $\hC \in \mC_i$.
\end{proposition}
\begin{proof}
  Each biset $\hC \in \mC_i$ is a minimal violated biset with respect
  to the graph $G_{\ell}[P_{i-1}]$. From Lemma~\ref{lem:aug-func-minimal} it follows
  that $C' \subseteq P_{i-1}$ and hence $\Gamma(\hC)$ can only contain
  vertices in $V \setminus P_{i-1}$.
\end{proof}

The lemma below shows that the algorithm maintains
dual feasibility with respect to the primal solution $P$.

\begin{lemma} \label{lem:complementary-slackness} The dual solution
  $y$ constructed by the primal-dual algorithm satisfies the
  primal complementary slackness conditions. More precisely, for each
  $v \in P$, $\sum_{\hS: v \in \Gamma(\hS)} y(\hS) = w(v)$.
\end{lemma}

\begin{proof}
  We prove the lemma by induction on the number of iterations of the
  first while loop. Initially, $y$ is zero and $P_0$ consists of all
  zero-weight vertices. Thus the complementary slackness conditions
  are satisfied at the beginning of the algorithm. Now consider
  iteration $i > 0$. From Proposition~\ref{prop:no-neighbor}, no vertex
  in $P_{i-1}$ is adjacent, in the graph $G_{\ell}$, to any biset in
  $\mC_i$. Thus increasing the dual variables corresponding to bisets
  of $\mC_i$ do not violate the tightness of vertices in $P_{i-1}$.
  And the
  only vertex added to $P_i$ in iteration $i$ is the one that becomes
  tight with respect to the dual increase in iteration $i$.
  Thus, at the end of iteration $i$, the required condition holds for
  $P_i$.
\end{proof}

\begin{lemma} \label{lem:primal-dual-poly-time} The primal-dual
  algorithm returns a feasible cover for $h_\ell$ in $G_{\ell}$
  when all vertices of $X_{\ell-1}$ are included in $P_0$.
\end{lemma}
\begin{proof}
  Proposition~\ref{prop:no-neighbor} and
  Lemma~\ref{lem:complementary-slackness} show that the algorithm
  inductively maintains the property that $\mC_i$ are the minimal
  violated sets with respect to the graph $G_{\ell}[P_{i-1}]$ and hence when
  the first while loop terminates we have the property that there are
  no minimal violated sets with respect to $G_{\ell}[P]$, and thus $G_{\ell}[P]$
  is a feasible cover. The reverse-delete step explicitly ensures that
  $G_{\ell}[Q]$ is a feasible cover.
\end{proof}

\subsection{Analysis of the Approximation Ratio}
\label{subsec:proof-main-thm}
We now analyze the approximation ratio. We will assume that the
function $h$ comes from an augmentation problem and that the algorithm
is well-defined and returns a feasible cover.
The basic lemma that underlies the primal-dual analysis is similar to
that of the edge-weighted case and relies on the uniform-growth
property of the dual variables.

\begin{lemma} \label{lem:elem-primal-dual} Let $Q$ be the set of
  vertices output by the primal-dual algorithm. Suppose that there
  exists a fixed value $\gamma$ such that, for each iteration $i$ of
  the primal-dual algorithm,
  $\sum_{\hC \in \mC_i} \card{Q \cap \Gamma(\hC)} \leq \gamma
  \card{\mC_i}$. Then $w(Q)$ is at most $\gamma$ times the value of an
  optimal solution of $\prob{Primal-LP}(G,\mP,h)$.
\end{lemma}
The content of the preceding lemma is the following. Consider the minimal
violated bisets in iteration $i$, $\mC_i$. Let
$Q_i = Q \setminus P_{i-1}$. From the reverse-delete step we can see
that $Q_i$ forms a {\em node-minimal} set that together with
$P_{i - 1}$ covers $h$. We are interested in $\gamma$, the ``average
degree''\footnote{Here we are abusing the term slightly and we refer
  to the ratio
  $\sum_{\hC \in \mC_i} \card{Q_i \cap \Gamma(\hC)} / \card{\mC_i}$ as
  the average degree of the bisets in $\mC_i$.} of the bisets in
$\mC_i$, with respect to nodes in $Q_i$. In general graphs, $\gamma$
can be $\Omega(n)$ in the worst case and we will not be able to prove
any reasonable guarantee on the performance of the primal-dual
algorithm.  However, planar graphs and more generally the graphs from
minor-closed families are sparse.  Thus we can bound the average
degree if we upper bound the number of nodes in $Q_i$ that are
neighbors of a biset in $\mC_i$. The following proof follows a
standard template in the context of primal-dual analysis but
we give it here for the sake of completeness.

\begin{proofof}{Lemma~\ref{lem:elem-primal-dual}}
	By Lemma~\ref{lem:complementary-slackness}, $y$ satisfies the
	primal complementary slackness conditions. Therefore we have
	\begin{align*}
		\sum_{v \in Q} w(v) = \sum_{v \in Q} \sum_{\hS\in \mP: v \in
		\Gamma(\hS)} y(\hS) = \sum_{\hS \in \mP} y(\hS) \card{Q \cap
		\Gamma(\hS)}.
	\end{align*}
	Note that, if we can show that $\sum_{\hS \in \mP} y(\hS) \card{Q
	\cap \Gamma(\hS)} \leq \gamma \sum_{\hS\in \mP} y(\hS) h(\hS)$,
	it will follow that we have a $\gamma$-approximation: since $y$
	is feasible, $\sum_{\hS\in \mP} y(\hS) h(\hS)$ is a lower bound on
	the fractional optimum, which in turn is a lower bound on the
	integral optimum.

	We show by induction on the number of iterations of the
	primal-dual algorithm that
	\begin{align*}
	\sum_{\hS\in \mP} y(\hS) \card{Q \cap \Gamma(\hS)} \leq
		\gamma \sum_{\hS\in\mP} y(\hS) h(\hS).
	\end{align*}
	Note that $y(\hS) > 0$ only if $h(\hS) = 1$. Therefore
	$\sum_{\hS\in\mP} y(\hS) h(\hS) = \sum_{\hS\in\sC} y(\hS)$ where $\sC$ is the collection of 
	violated bisets with respect to $G$. 
	In the primal-dual algorithm, we only increase the dual values of violated bistes and thus it
	suffices to prove that
	\begin{align*}
	\sum_{\hS \in \sC} y(\hS) \card{Q \cap \Gamma(\hS)} \leq
		\gamma \sum_{\hS \in \sC} y(\hS).
	\end{align*}
	Initially, all dual variables $y(\hS)$ are zero and therefore the
	inequality holds. Now consider iteration $i$ of the primal-dual
	algorithm. Recall that $P_{i - 1}$ is the set of all vertices
	selected in the first $i - 1$ iterations of the primal-dual
	algorithm and $\mC_i$ is the set of all minimal violated bisets
	with respect to $G_{\ell}[P_{i - 1}]$.

	Let $\epsilon$ denote the amount by which we increased $y(\hS)$
	for $\hS \in \mC_i$ in iteration $i$. The left-hand side
	increases by $\sum_{\hC \in \mC_i} \epsilon \card{Q \cap
	\Gamma(\hC)}$, and the right-hand side increases by $\gamma
	\epsilon \card{\mC_i}$.  Therefore it suffices to show that
	\begin{align*}
	\sum_{\hC \in \mC_i} \card{Q \cap \Gamma(\hC)} \leq \gamma
		\card{\mC_i}.
	\end{align*}
	Recall that $Q_i = Q \setminus P_{i - 1}$. By
	Proposition~\ref{prop:no-neighbor}, for each $\hC
	\in \mC_i$, $\Gamma(\hC) \cap P_{i - 1}$ is empty and thus
	$\Gamma(\hC) \cap Q = \Gamma(\hC) \cap Q_i$.  Therefore we
	can rewrite the inequality above as:
	\begin{align*}
	\sum_{\hC \in \mC_i} \card{Q_i \cap \Gamma(\hC)} \leq
		\gamma \card{\mC_i},
	\end{align*}
\noindent
which holds by the assumption in the statement of the lemma.
\end{proofof}

The key technical contribution of the paper is to bound $\gamma$
and it is captured by the following theorem whose
proof is in Section~\ref{sec:main-thm}.

\begin{theorem}
  \label{thm:main-counting}
  Consider phase $\ell$ of the augmentation algorithm.
  For $X \supset X_{i-1}$ let $H_X = (V, E[X])$ be a subgraph
  of $G_\ell$. Let $\mC$ be the collection of minimal violated bisets
  of $h_\ell$ with respect to $H_X$. Suppose $Q \subseteq V \setminus X$ is a
  node-minimal set such that $G_\ell[X \cup Q]$ is a feasible cover for $h_\ell$.
  Then $\card{Q \cap (\bigcup_{\hC \in \mC} \Gamma(\hC))} \leq 4 \card{\mC}$.
\end{theorem}

Exploiting the sparsity of planar graphs and more generally minor-closed
families of graphs together with the preceding theorem, we obtain the
following lemma.

\begin{lemma} \label{lem:elem-counting-lemma}
  Consider an instance of \prob{Elem-SNDP} over a graph $G$ that
  belongs to a proper minor-closed family of graphs ${\cal G}$.
  Suppose we run the primal-dual algorithm in phase $\ell$
  to cover $h=h_\ell$. Let
  $Q_i = Q \setminus P_{i - 1}$. Then,
  $\sum_{\hC \in \mC_i} \card{Q_i \cap \Gamma(\hC)} \leq c
  \card{\mC_i}$, where $c$ is a constant that depends only on the
  family ${\cal G}$.  In particular, if $G$ is a planar graph then
  $\sum_{\hC \in \mC_i} \card{Q_i \cap \Gamma(\hC)} \leq 10
  \card{\mC_i}$.
\end{lemma}
\begin{proof}
  Let $N_i = Q_i \cap (\bigcup_{\hC \in \mC_i} \Gamma(\hC))$.
  By Theorem~\ref{thm:main-counting},
  $\card{N_i} \leq 4 \card{\mC_i}$. Let $\mC_i =
  \{\hC_1,\hC_2,\ldots,\hC_r\}$. Since $\mC_i$ is a collection of
  minimal violated bisets of biuncrossable function $h_\ell$
  they do not overlap which means that $C_1,C_2,\ldots,C_r$ are
  pairwise disjoint sets. We also have the property that $N_i \cap C_j
  = \emptyset$ for $1 \le j \le r$.
   Further, from Lemma~\ref{lem:aug-func-minimal}, $G[C_j]$ is connected for
  $1 \le j \le r$.  Next, we construct a minor $K$ of $G$ as follows. Let $V' =
  (\bigcup_{\hC \in \mC_i} C ) \cup N_i$. Note that we do not
  include the boundary vertices of the bisets of $\mC_i$ in
  $V'$. We start with $K = G[V']$. For each biset $\hC_j \in
  \mC_i$, we shrink the set $C_j$ to a single vertex $v_j$. We also
  remove parallel edges in order to get a simple graph. The
  resulting graph is indeed a minor of $G$ since (i) $G[C_j]$ is
  connected for each $j$, and (ii) $C_1,C_2,\ldots,C_r,N_i$ are
  pairwise disjoint. The total number of nodes in $K$ is
  $\card{N_i} + \card{\mC_i} \le 5 \card{\mC_i}$. In $K$ we also
  remove edges between two nodes of $N_i$ which results in a bipartite
  graph with $N_i$ on one side, and the vertices $v_1,v_2,\ldots,v_r$
  corresponding to $C_1,\ldots,C_r$ on the other side.
  Note that $K$ is still a minor of $G$.

  Recall that we have $Q_i \cap \Gamma(\hC) \subseteq N_i$ for each
  $\hC \in \mC_i$. Therefore
  $\sum_{\hC \in \mC_i} \card{Q_i \cap \Gamma(\hC)}$ is equal to the
  number of edges in the bipartite graph $K$. Since $K$ is from a
  minor-closed family ${\cal G}$, from \cite{Kostochka84} it follows
  that there is a constant $c'$ that depends only on the family such
  that $|E(K)| \le c' |V(K)| \le 5c' |\mC_i|$.
  Suppose $G$ is a planar graph. Then $K$ is a bipartite planar
  graph and in this case it is well-known that $\card{E(K)} \le 2
  \card{V(K)} \le 10 \card{\mC_i}$.
\end{proof}

\subsection{Proof of Theorem~\ref{thm:main-counting}}
\label{sec:main-thm}
In this section, we prove Theorem~\ref{thm:main-counting} using a
counting argument which is a generalization of the counting argument
of~\cite{ChekuriEV12}.
We use $H$ in place of $H_X$ to simplify notation.
We use $K$ to denote the graph $(V(G), E(G) \setminus E(H))$.
Consider the biset function $h'$ where $h'(\hS) = 1$ iff
$h_\ell(\hS) = 1$ and $\delta_{H}(\hS) = \emptyset$. By
Proposition~\ref{prop:difference-bisup} and
Lemma~\ref{lem:delta-bisubmod}, $h'$ is a $\{0,1\}$-biuncrossable
function. Note that $h'$ is the residual function of $h_\ell$
in the graph $H$. Let $Q' = Q \cup X$. Recall that $Q$ a is
node-minimal set such that $G_\ell[Q']$ covers $h_{\ell}$.
Equivalently this means that $Q$ is a node-minimal set such
that $K[Q \cup X]$ covers $h'$.

The main idea in the proof is to pick a subset $M$ of the edges of
$K[Q']$ such that $M$ is an \emph{edge-minimal} feasible cover for
$h'$. An edge-minimal set allows us to use an approach that was
introduced by Williamson \etal \cite{WilliamsonGMV95} for the
edge-weighted SNDP problem. More precisely, for each edge $e \in M$,
we can pick a \emph{witness biset} that is a violated biset of $h'$
such that $e$ is the only edge of $M$ that is leaving the
biset. Moreover, we can pick a laminar family of witness bisets for
all edges in $M$ that allows us to upper bound the number of edges in
$M$ incident to the components of $\mC$ in terms of $\card{\mC}$.

Since $K[Q']$ is a node-minimal cover of $h'$ and not an edge-minimal
cover, it is possible that there is a vertex $u \in Q$ connected to a
component of $\mC$ in $K[Q']$ but none of the edges connecting $u$ to
components of $\mC$ are in $M$.  Thus we cannot use the family of
witness bisets of an edge-minimal cover of $h'$ to bound the number of
these vertices. We address this issue by counting these vertices
separately using a witness family of bisets for a different set of
\emph{non-redundant} edges.

We refer to the vertices in
$Q \cap \big(\bigcup_{\hC \in\mC} \Gamma_G(\hC)\big)$ as
\textbf{critical} vertices; these are the vertices of $Q$ that are
adjacent to at least one biset in $\mC$ and the goal is to show
that there are at most $4\card{\mC}$ of them.  We refer to the edges in
$\cup_{\hC \in \mC} \delta_{K}(\hC)$ as \textbf{red edges}, and all
other edges of $K$ as \textbf{blue edges}. Every critical vertex
is incident to at least one red edge.
We define two subsets of edges $F$ and $F'$ below.

We start with $F = E(K)$ and we remove some of the edges as follows.
We consider the \emph{blue edges} in an arbitrary order. Let $e$ be
the current edge. If $F \setminus \set{e}$ is a feasible solution for
$h'$, we remove $e$ from $F$. This procedure gives us a set of edges
in which each blue edge is necessary, in the sense that removing any
blue edge from $F$ will make it an infeasible cover for $h'$. As we
will see shortly, we can use the blue edges in $F$ to upper bound the
number of critical vertices that are incident to at least one blue
edge of $F$.  We refer to critical vertices that are incident to a
blue edge of $F$ as \textbf{regular} vertices, and we refer to all
other critical vertices as \textbf{special} vertices.

In order to count the special vertices, we pick a subset $F'$ of $F$
as follows. We start with $F' = F$ and consider the \emph{red edges}
of $F'$ in some arbitrary order. Let $e$ be the current edge.  If
$F' \setminus \set{e}$ is a feasible cover of $h'$, we remove $e$ from
$F'$. We can use the red edges in $F'$ to upper bound the remaining
critical vertices. Since $Q$ is a node-minimal cover for $h'$, each
special vertex is incident to at least one red edge of $F'$.

We consider the regular and special vertices
separately. Theorem~\ref{thm:main-counting} follows from the following
lemmas.

\begin{lemma} \label{lem:pair-regular}
  The number of regular vertices is at most $2\card{\mC}$.
\end{lemma}

\begin{lemma} \label{lem:pair-special}
  The number of special vertices is at most $2\card{\mC}$.
\end{lemma}

First, we formally define the notion of \emph{witness bisets} and
discuss the existence of a family of \emph{non-overlapping} bisets
that is a key part in the counting argument.

\begin{definition}[\prob{Witness Biset}]
  Let $G=(V,E)$ be an input graph and let $h$ be a 
  $\{0,1\}$-biuncrossable function defined on a crossing bifamily
  $\mP \subseteq 2^V \times 2^V$.  Let $F \subseteq E$ be a feasible cover of
  $h$. Then, $\hS_e$ is an $F$-witness biset of $e \in F$ iff
  $h(\hS_e) = 1$ and $\delta_F(\hS_e) = \set{e}$.
\end{definition}

\begin{definition}[\prob{Laminar bifamily}]
  A family of bisets $\mF$ is laminar iff for any $\hS, \hT\in \mF$,
  $\hS$ and $\hT$ are non-overlapping.
\end{definition}

Given a feasible cover $F$ for a requirement function $h$ we say that
an edge $e \in F$ is \emph{non-redundant} if $F \setminus \{e\}$ is
not a cover. A set $M \subseteq F$ is non-redundant if each edge
$e \in M$ is non-redundant.  The following lemma is known from past
work and we provide a proof for the sake of completeness in the
appendix.

\begin{lemma} \label{lem:pair-laminar-witness-family} Let $F$ be a
  feasible cover of a $0$-$1$ bi-uncrossable function $h$. Let
  $M\subseteq F$ be a set of non-redundant edges with respect to $F$.
  There exists a laminar family of bisets
  $\mL = \set{\hS_e \sep e \in M}$ such that $\hS_e$ is an $F$-witness
  biset for $e$.
\end{lemma}

Our approach is to use laminar families of witness bisets for the blue
edges of $F$ and the red edges of $F'$ in order to count the regular
and special vertices. Before we turn our attention to the counting
arguments, we describe some properties of laminar families of witness
bisets that we will need.

\medskip
\noindent {\bf Laminar witness tree.} For a \emph{laminar} collection
of bisets $\mL$ defined on set $V$ let $\mL^+$ denote the extended
laminar family of $\mL$, $\mL^+:=\mL \cup \set{(V,V)}$. We associate a
tree $\mT_{\mL^+}$ with the family $\mL^+$ as follows. The tree
$\mT_{\mL^+}$ has a node $\nu_{\hS}$ for each biset $\hS \in
\mL^+$. For any two bisets $\hS$ and $\hT$ of $\mL$ such that
$\hS \subset \hT$, we add an edge from the node of $\mT$ representing
$\hS$ to the node representing $\hT$ iff there is no biset
$\hX \in \mL$ such that $\hT \subset \hX \subset \hS$. We view the
tree $\mT_{\mL^+}$ as a rooted tree whose root is the node
corresponding to the biset $(V, V)$.

In the following, we consider a biuncrossbale function
$h: \mP \rightarrow \set{0, 1}$, a cover $F$ for $h$, and a set of
non-redundant edges $M \subseteq F$.  We also fix a laminar family $\mL$ of
$F$-witness bisets for $M$, and denote the tree associated with
$\mL^+$ by $\mT$.

\begin{definition}
  A biset $\hS \in \mL^+$ \textbf{owns} $u\in V$ iff $\hS$ is the
  minimal biset in $\mL^+$ that contains
  $u$ in its inner part.
\end{definition}

\begin{prop} \label{prop:own-unique}
  For each vertex $u \in V$, there is a unique biset in $\mL^+$ that owns $u$.
\end{prop}
\begin{proof}
  Note that $(V, V)$ contains $u$ in its inner part, and thus there
  is a biset in $\mL \cup \set{(V, V)}$ that owns $u$.
  Suppose for contradiction that two distinct bisets $\hX$ and $\hY$
  of $\mL \cup \set{(V, V)}$ own $u$.
  Note that the set $X \cap Y$ is non-empty, since $u \in X \cap
  Y$. Since $\hX$ and $\hY$ do not overlap, we must have $\hX
  \subseteq \hY$ or $\hY \subseteq \hX$.  Therefore one of $\hX, \hY$
  does not own $u$, which is a contradiction.
\end{proof}

\begin{prop} \label{prop:minimal-violated-own}
  Let $\hC$ be a minimal violated biset. Then
  all vertices of $C$ are owned by the same biset in $\mL^+$.
\end{prop}

\begin{proof}
  Let $\hC$ be a minimal violated biset of $h$. If $\hC \in \mL^+$ we
  are done. Consider any $\hY \in \mL$; it is a violated biset.
  By Lemma~\ref{lem:minimal-violated-crossing}, $\hC$ and
  $\hY$ do not overlap. Thus either
  $C \subseteq Y$ or $C \cap Y$ is empty.
  Consider the minimal biset $\hY \in \mL^+$ such that $\hC \subset
  \hY$ (since $\hC \subset (V,V)$ such a biset exists).
  Then $\hY$ owns all vertices of $C$.
\end{proof}

It is convenient to abuse the notation and say that the node
$\nu_{\hS}$ of $\mT$ owns $u$ if $\hS$ owns $u$. Additionally, we say
that $\hS$ owns $\hC$ if it owns the inner part of $\hC$.
Consider an edge $e = uv$ in $M$, and assume that $\hS, \hT \in \mL^+$
own $u$ and $v$, respectively. The lemma below shows that
either $\nu_{\hS}$ is an ancestor of $\nu_{\hT}$ or vice-versa.
It is possible for $\hT$ to be a proper ancestor of $\hS$ (that is,
there is another biset $\hY$ in the family such that
$\hS \subset \hY \subset \hT$) or vice-versa.

\begin{lemma} \label{lem:pair-witness} Let $e = uv$
  be an edge of $M$. Let $\hS_e$ be the $F$-witness biset of $e$ in
  $\mL$ and suppose $\hS,\hT\in \mL^+$ own $u$ and $v$, respectively.
  \begin{itemize}
  \item Then $\hS_e = \hS$ and $\hS \subset \hT$, or $\hS_e = \hT$ and
    $\hT \subset \hS$.
  \item Suppose, in addition, $u$ is not contained in the boundary of
    any bisets in $\mL^+$. If $\hT \subset \hS$ then $\hT$ is a child
    of $\hS$, that is, there is no $\hY \in \mL^+$ such that $\hT
    \subset \hY \subset \hS$.
  \end{itemize}
\end{lemma}
\begin{proof}
  We consider the first part.
  Let $w$ be the endpoint of $e$ that is contained in the inner part
  of $\hS_e$. Suppose for contradiction that $\hS_e$ does not own
  $w$. Then there exists a biset $\hX \subset \hS_e$ in $\mL$ such
  that $w\in X$.  Note that $e \in \delta_F(\hX)$. However, $\hX$ is
  an $F$-witness biset for an edge in $M \setminus \set{e}$, which is
  a contradiction. Therefore, $\hS_e$ owns $w$. Thus $\hS_e = \hS$ or
  $\hS_e = \hT$.  Without loss of generality let us assume that $\hS_e = \hS$ and that
  $u \in S$. To complete the proof, we need to show that
  $\hS \subset \hT$. Note that we may assume that $\hT \neq (V, V)$;
  otherwise, $\hS \subset \hT$ trivially holds.  Moreover,
  $\hS \neq \hT$, since $v$ is in $V \setminus S'$ and $v\in T$.
  Since $\hS$ is an $F$-witness biset for $e$, it follows that $\hT$
  is an $F$-witness biset for an edge of $M \setminus \set{e}$.
  Therefore $e \notin \delta_F(\hT)$ and, since $v \in T$, we must
  have $u \in T'$. Since $\hS$ and $\hT$ do not overlap and
  $u \in S \cap T'$, either $\hS \subset \hT$ or $\hT \subset
  \hS$. However, since $v \in T \setminus S'$, we cannot have
  $\hT \subset \hS$. Thus, $\hS \subset \hT$.

  We now consider the second part where we assume that $u$ is not in
  the boundary of any biset of $\mL^+$. If $\hT \subset \hS$ then from
  the preceding part we have $\hS_e = \hT$ and $v \in T$.  Suppose for
  contradiction that there exists a biset
  $\hY  \in \mL \setminus \set{\hS, \hT}$ such that
  $\hT \subset \hY \subset \hS$. Since $\hY$ is an $F$-witness biset
  for an edge of $M \setminus \set{e}$, we must have
  $e \notin \delta_F(\hY)$.  Therefore, since $u\in S\setminus T$ and
  $v\in T$, $u\in Y'$; since $u$ is owned by $\hS$ and not by $\hY$ it
  implies that $u \in \bd(\hY)$.  However, by assumption
  $u$ is not in the boundary of any
  biset of $\mL^+$ which implies there is no such $\hY$.
\end{proof}

The following lemma is an important one that underlies the analysis.

\begin{lemma} \label{lem:pair-non-witness-edge-mapping} Let $e= uv$ be
  an edge of $F \setminus M$ such that $e\in \delta_F(\hC)$ for a
  minimal biset $\hC$.  Let $u$ be the endpoint of $e$ that is in
  $V \setminus C'$. Let $\hS$ be the biset of $\mL^+$ that owns
  $\hC$. If $u$ is not contained in the boundary of any biset of
  $\mL^+$, then $\hS$ owns $u$.
\end{lemma}
\begin{proof}
  First we claim that $u \in S$. If $\hS = (V, V)$, the claim holds
  trivially.  Therefore we may assume that $\hS\in \mL$. Since
  $v \in C$ and $C \subset S$ we have $v \in S$. If $u \not \in S'$
  then $e \in \delta(\hS)$ but $\hS$ is an $F$-witness biset for some
  $e' \in M$ and $e \neq e'$ which is a contradiction.

  Now suppose for contradiction that $\hS$ does not own $u$.  Then
  there is a biset $\hT \in \mL$ such that $\hT\subset \hS$, and
  $u \in T$.  Since $\hT$ is an $F$-witness biset for an edge in $M$
  and $u \in T$, we must have $v\in T'$ which implies that
  $C\cap T' \neq \emptyset$. By
  Lemma~\ref{lem:minimal-violated-crossing}, $\hC$ and $\hT$ are
  non-overlapping, and therefore we must have $C \subseteq T$ which
  contradicts the fact that $\hS$ owns $\hC$. Hence $\hS$ owns $u$.
\end{proof}

\begin{prop} \label{prop:pair-witness-leaf}
  If $\nu_{\hS}$ is a leaf of $\mT$, then $\hS$ owns a minimal violated biset.
\end{prop}
\begin{proof}
  Since $\hS$ is a violated biset with respect to $h$, there exists a
  minimal violated biset $\hC \in \mC$ such that $\hC \subseteq
  \hS$. Moreover, $\hS$ is a minimal biset of $\mL^+$; thus, $\hS$ owns $\hC$.
\end{proof}

\noindent
\textbf{Bijection between the edges of $M$ and $\mT$.}
We define the following bijection between the edges of $\mT$
and the edges of $M$. Let $e$ be an edge of $M$ and let $\hS_e$ be
the witness biset for $e$. The node $\nu_{\hS_e}$ has a parent $\nu_{\hT}$
in $\mT$, and we associate $e$ with the edge
$(\nu_{\hT}, \nu_{\hS_e})$ in $\mT$. We say that the edge $e$
\emph{corresponds} to the edge $(\nu_{\hT}, \nu_{\hS_e})$.

\begin{prop} \label{prop:deg-bijection}
  Let $\hC$ be a minimal violated biset that is owned by $\nu_{\hS}\in\mT$.
  Then each edge $e\in \delta_{M}(\hC)$ whose endpoints are not in the
  boundary of any violated bisets, corresponds to an edge of
  $\mT$ incident to $\nu_{\hS}$.
\end{prop}
\begin{proof}
  Let $e=uv$ be an edge in $\delta_M(\hC)$ and let $u \in C$ and
  $v \not \in C'$. Since $\hS$ owns $\hC$, $u$ is owned by $\hS$.  If
  $v \in V \setminus S'$ then $\{e\} = \delta_M(\hS)$ which implies that
  $\hS = \hS_e$ and in this case $e$ corresponds to
  $(\nu_{\hS},\nu_{\hT})$ where $\hT$ is the parent of
  $\hS$. Otherwise $v \in S'$ and since $e \in M$ it follows
  that $\{e \} = \delta_M(\hT)$ for some descendent of $\hS$.
  Since $u$ is not on the boundary of any violated biset by assumption,
  from the second part of Lemma~\ref{lem:pair-witness},
  $\hT$ is a child of $\hS$ which implies that $e$ corresponds
  to $(\nu_{\hT},\nu_{\hS})$ which is incident to $\nu_{\hS}$.
\end{proof}

\noindent
\textbf{Counting argument for regular vertices.}
Let $\mLB = \set{\hS_e \sep \text{$e$ is a blue edge in $F$}}$
be a laminar family of $F$-witness bisets of the blue edges in $F$
that is guaranteed by Lemma~\ref{lem:pair-laminar-witness-family}.
Let $\mTB$ be the tree associated with $\mLB^+$.

Recall that each regular vertex is incident to at least one blue edge
of $F$. Additionally, recall that $F$ contains all the red edges of
$K[Q']$. Therefore, for each regular vertex $u$, there is a red edge
in $F$ that is incident to $u$. Moreover, by
Lemma~\ref{lem:aug-func-bd}, no violated biset
contains a critical vertex in its boundary.

We charge each regular vertex $u$ as follows. Let $\hC \in \mC$ be a
minimal violated biset for which there exists a red edge
$wu \in \delta_G(\hC)$ such that $w \in C$ and $u \in V \setminus
C'$. Moreover, let $e = uv$ be a blue edge of $F$ and let
$\hS_e \in \mLB$ be the $F$-witness biset for $e$. Suppose that $\hS$
and $\hT$ be the bisets that own $u$ and $v$, respectively.

By Lemma~\ref{lem:pair-non-witness-edge-mapping} and the fact that $u$
is not on the boundary of any violated biset, $\hS$ owns $\hC$.
Additionally, by Lemma~\ref{lem:pair-witness}, we have
either $\hS \subset \hT$ or $\hT \subset \hS$, and we consider each of
these cases separately.  Suppose that $\hS \subset \hT$. It follows
from Lemma~\ref{lem:pair-witness} that $\hS = \hS_e$. We
charge $u$ to $\hC$.  We refer to such a charge as a \emph{parent
  charge}.  Next, suppose that $\hT \subset \hS$. By
Lemma~\ref{lem:pair-witness} and the fact that $u$ is not on
the boundary of any violated biset, $\nu_{\hT}$ is a child of
$\nu_{\hS}$. Since, by Proposition~\ref{prop:pair-witness-leaf}, each
leaf of $\mTB$ owns a biset of $\mC$, there is a descendant of
$\nu_{\hT}$ (possibly $\nu_{\hT}$ itself) that owns a biset of
$\mC$. Let $\nu_{\hX}$ be the closest such descendant, i.e., a
descendant whose distance to $\nu_{\hT}$ is minimized. (If there are
several descendants whose distance to $\nu_{\hT}$ is minimum, we pick
one of them arbitrarily.) We charge $u$ to one of the bisets of $\mC$
that $\nu_{\hX}$ owns. We refer to this charge as a \emph{subtree
  charge}, since $u$ is charged in a subtree rooted at a child of the
node $\nu_{\hS}$ that owns $u$.

\begin{lemma}\label{lem:parent-charge}
	There is at most one parent charge to each biset $\hC \in
	\mC$.
\end{lemma}
\begin{proof}
	Let $\hS $ be the biset of $\mLB^+$
	that owns $\hC$. Suppose that $\hC$ incurs a parent charge from a
	vertex $u$. Then $u$ is in $S$ and there is a blue edge $e = uv \in F$
	such that $\hS =\hS_e$. Since $\hS$ is an $F$-witness biset for exactly one
	edge, there is at most one parent charge to $\hC$.
\end{proof}

\begin{lemma}\label{lem:subtree-charge}
	There is at most one subtree charge to each biset $\hC \in \mC$.
\end{lemma}
\begin{proof}
	Let $\hX$ be the biset of $\mLB^+$
	that owns $\hC$. Suppose that there is a descendant charge to
	$\hC$ corresponding to a vertex $u_1$. Then there is a blue edge
	$e_1 = u_1v_1$ whose endpoints are owned by $\hS_1$
	and $\hT_1$ (respectively) and $\nu_{\hT_1}$ is a child of $\nu_{\hS_1}$.
	By the way we choose $\hX$, $\nu_{\hX}$ and $\nu_{\hS_1}$ are the only
	nodes on the path in $\mTB$ from $\nu_{\hS_1}$ to $\nu_{\hX}$ that
	own a biset of $\mC$. Moreover, $\hT_1$ is the witness biset of $e_1$.

	Suppose for contradiction that there is another descendant charge to
	$\hC$ from $u_2 \neq u_1$. Let $\hS_2$ and
	$\hT_2$ be the bisets of $\mL^+$ that own
	$u_2$ and $v_2$ (respectively).  By the same argument as above,
	$\nu_{\hS_2}$ is a child of $\nu_{\hT_2}$ and there is no node in the path from
	$\nu_{\hS_2}$ to $\nu_{\hX}$ other than its endpoints that owns a biset of $\mC$.
	Moreover, $\hT_2$ is the witness biset of $e_2=u_2v_2$ (which is different from $e_1$).


	Since $\hT_1$ and $\hT_2$ are distinct bisets of $\mLB^+$ that both contain
	$\hX$, either $\hT_1 \subset \hT_2$ or $\hT_2 \subset \hT_1$. Moreover,
	since for each $i\in \set{1,2}$, $\hX \subset \hS_{i}$, $\hS_i \not\subset \hT_{2-i}$
	and $\nu_{T_i}$ is a child of $\nu_{S_i}$, we must have $\hS_1 = \hS_2$.
	Thus $\nu_{\hT_1}$ and $\nu_{\hT_2}$ are children of $\hS$ which contradicts the
	fact that either $\hT_1 \subset \hT_2$ or $\hT_2 \subset \hT_1$.
	Hence there is at most one subtree charge to each minimal violated biset $\hC$.
\end{proof}

\begin{proofof}{Lemma~\ref{lem:pair-regular}}
	By Lemma~\ref{lem:parent-charge} and~\ref{lem:subtree-charge},
	each biset of $\mC$ is charged at most twice and thus the
	number of regular vertices is at most $2\card{\mC}$.
\end{proofof}

\medskip\noindent \textbf{Counting argument for special vertices.}
Recall that $F'$ is an edge-minimal cover of $h'$. Moreover, a
critical vertex $v$ is special only if there is a red edge $e$
incident to $v$ such that $e \in \delta_{F'}(\hC)$ for
a minimal violated biset $\hC$ where $v\in V\setminus C'$. Thus the
total number of special vertices is upper bounded by
$\sum_{\hC \in \mC} \card{\delta_{F'}(\hC)}$.  Next we adopt the
argument of Jain \etal \cite{JainMVW02} to show that, for any
edge-minimal cover $F'$, $\sum_{\hC \in \mC} \card{\delta_{F'}(\hC)}$
is at most $2\card{\mC}$. Let $\mLR$ be the laminar family of witness
bisets of red edges with respect to $F'$ and let $\mTR$ be the tree
representation of $\mLR^+$. Let $\mA$ denotes the set of vertices in
$\mTR$ that own a minimal violated biset. By
Proposition~\ref{prop:deg-bijection},
$\sum_{\hC\in\mC}\card{\delta_{F'}(\hC)} \leq \sum_{\nu\in \mA}
\deg(\nu)$.  Here $\deg(\nu)$ refers to the degree of node $\nu$
in $\mTR$. Note that by Proposition~\ref{prop:pair-witness-leaf},
all leaf vertices are in $\mA$ and there is at most one node in
$V(\mTR)\setminus\mA$ with degree less than $2$; root node.  Hence,
\begin{align*}
\sum_{\hC\in\mC}\card{\delta_{F'}(\hC)} &\leq \sum_{\nu\in \mA} \deg(\nu) \\
	&\leq  \sum_{\nu\in V(\mTR)} \deg(\nu) - \sum_{\nu\in V(\mTR)\setminus \mA} \deg(\nu)\\
	&\leq 2(\card{V(\mTR)}-1) - 2(\card{V(\mTR)}-\card{\mA}-1)\\
	&\leq 2\card{\mA} \leq 2\card{\mC}.
\end{align*}

Thus we can upper bound the number of special vertices
by $2\card{\mC}$ which proves Lemma~\ref{lem:pair-special}. We remark
that some of the regular vertices are counted in this step as well,
but this can only help us.

\section{Algorithm for $\set{0, 1, 2}$ VC-SNDP}
\label{sec:node-wt-vc-sndp}

In this section, we prove the following theorem.
family of graphs.

\begin{theorem} \label{thm:vc-sndp}
  There is an $O(1)$-approximation for node-weighted
  \prob{VC-SNDP} when the requirements are in $\set{0, 1, 2}$
  and the input graph belongs to a proper minor-closed family of graphs.
\end{theorem}

\noindent
We construct a solution in two stages. In the first stage we use an
algorithm for node-weighted Steiner forest to find a set
$X_1 \subseteq V$ of nodes such that $G[X]$ connects each pair $uv$
with $r(u,v) \ge 1$. A constant factor approximation for this in
proper minor-closed families of graphs follows from prior work that we
already discussed \cite{DemaineHK14,Moldenhauer13}.  Letting $\opt$
denote the weight of an optimum solution for the initial instance we
see that $w(X_1) = O(1) \opt$. Let $F$ be the edge set of the graph
$E[X_1]$.  In the second stage, we augment $X_1$ to $2$-connect pairs
$(s,t)$ with connectivity requirement $2$ that . For the second stage,
as with \prob{Elem-SNDP}, we define a $\{0,1\}$-biuncrossable function
$h$ and a graph $G' = (V, E \setminus F)$.  Let $\mPV$ be the
collection of all bisets over $V$.  The requirement function
$\rV: \mPV \rightarrow \set{0,1,2}$ for each biset $\hS$ is defined as
the maximum connectivity requirement over all pair of vertices that
are separated by $\hS$.  Let $h: \mPV \rightarrow \set{0, 1}$ be the
function such that $h(\hS) = 1$ iff $\rV(\hS)=2$ and
$\card{\delta_F(\hS)} +\card{\bd(\hS)} =1$.  By Menger's theorem on
vertex connectivity, a feasible cover of $\hV$ together with $F$ is a
feasible solution for the \prob{VC-SNDP} instance.
For the second stage we are only interested in those pairs $(s,t)$
such that $r(s,t) = 2$ while $s$ and $t$ are only $1$-connected in
$G[X_1]$. We call a vertex $u$ a terminal for the second stage
if it participates in such a pair.

\begin{proposition}[Lemma 5.1 in~\cite{FleischerJW06}]\label{prop:vc-supermodular}
  The function $\rV$ is biuncrossable.
\end{proposition}
\begin{proposition}\label{prop:vc-bimaximal}
  The funtion $\rV$ is bimaximal.
\end{proposition}
\begin{proof}
Let $\hS, \hT \in \mPV$ and let $(s, t)$ be a pair
of terminals that have the maximum connectivity requirement among all
terminal pairs separated by $\hS \cup \hT$, i.e., $\rV(\hS \cup \hT) =
r(s, t)$. Since $s \in S \cup T$, we have $s \in S$ or $s \in T$;
without loss of generality, assume $s \in S$. Since $t \in V \setminus (S' \cup T')$,
the pair $(s, t)$ is separated by $\hS$ and thus $\rV(\hS \cup \hT) \leq
\rV(\hS) \leq \max\set{\rV(\hS), \rV(\hT)}$.
\end{proof}
\begin{lemma}\label{lem:vertex-uncrossable}
	The function $h$ is biuncrossable.
\end{lemma}
\begin{proof}
By Lemma~\ref{lem:delta-bisubmod}, $\card{\delta_F(.)}$ is bisubmodular and by Proposition~\ref{prop:boundary-bisubmodular}, $\card{\bd(.)}$ is bisubmodular. If $h(\hS) = h(\hT) =1$ then
$\rV(\hS) = \rV(\hT) = 2$, and $\card{\delta_{F}(\hS)} + \card{\bd(\hS)} = \card{\delta_{F}(\hT)} + \card{\bd(\hT)} = 1$.

Since $\rV$ is biuncrossable (Proposition~\ref{prop:vc-supermodular}),
\begin{align*}
\rV(\hS) + \rV(\hT) \leq \max\set{\rV(\hS \cap \hT) + \rV(\hS \cup \hT), \rV(\hS \setminus \hT) + \rV(\hT \setminus \hS)}.
\end{align*}
WLOG, assume that $\rV(\hS) + \rV(\hT) \leq \rV(\hS \cap \hT) + \rV(\hS \cup \hT)$ which by the upper bound of $2$ on the connectivity requirements implies that $\rV(\hS \cap \hT) = \rV(\hS \cup \hT) = 2$.

Moreover, since $\card{\delta_F(.)}$ and $\card{\bd(.)}$ are both bisubmodular, $\card{\delta_F(\hS)} + \card{\bd(\hS)} + \card{\delta_F(\hT)} + \card{\bd(\hT)} \geq \card{\delta_F(\hS\cap\hT)} + \card{\bd(\hS\cap\hT)} + \card{\delta_F(\hS\cup\hT)} + \card{\bd(\hS\cup\hT)}$. Since the edge set $F$ connects all pair of terminals with non zero connectivity requirements, both $\card{\delta_F(\hS\cap\hT)} + \card{\bd(\hS\cap\hT)},  \card{\delta_F(\hS\cup\hT)} + \card{\bd(\hS\cup\hT)}$ are at least $1$. Hence, $\card{\delta_F(\hS\cap\hT)} + \card{\bd(\hS\cap\hT)} = \card{\delta_F(\hS\cup\hT)} + \card{\bd(\hS\cup\hT)} = 1$.

The other case holds similarly and thus $h$ is biuncrossable.
\end{proof}

Next, we prove an analogues of Lemmas~\ref{lem:aug-func-bd} and
\ref{lem:aug-func-minimal}
which show that $h$ satisfies the key properties that
allowed us to use and analyze the primal-dual algorithm
from Section~\ref{sec:main-result}.

\begin{lemma}
  \label{lem:aug-func-bd-vc}
  For any $X \supset X_1$ let $H'_X = (V, E_{G'}[X])$ be a subgraph
  of $G'$. Suppose $\hC$ is a violated biset
  of $H'_X$ with respect to $h$. Then
  $\bd(\hC) \subseteq X$.
\end{lemma}
\begin{proof}
  We use $H'$ in place of $H'_X$ for ease of notation.  Suppose $\hC$
  is a violated biset with respect to $h$ in $H'$ and there is a
  vertex $u \in \bd(\hC)$ such that $u \not \in X$.  Consider the
  biset $\hC_1 = (C, \bd(\hC) \setminus \{u\})$. Note that $u$ is not
  a terminal and hence $\rV(\hC_1) = \rV(\hC) = 2$.  Since $\hC$ is a
  violated biset in $H'$ we have
  $\card{\bd(\hC)} + \delta_{H'}(\hC) = 1$ and since
  $\card{\bd(\hC)} \ge 1$ (since $u$ is in the boundary) we have
  $\delta_{H'}(\hC) = \emptyset$. Consider the bisets $\hC$ and
  $\hC_1$ in the graph $H_1 = (V, E[X_1])$. Since $\rV(\hC) = 2$
  and $\hC$ is violated in $H'$,
  we have $\card{\bd(\hC)} + \card{\delta_{H_1}(\hC)} = 1$ but then
  $\card{\bd(\hC_1)} + \card{\delta_{H_1}(\hC_1)} = 0$ since
  $u$ has no edges incident to it in $H_1$. Since $\rV(\hC_1) = 2$
  this implies that $H_1$ is not a feasible solution to $1$-connect
  the terminals in the first stage, a contradiction.
\end{proof}

\begin{lemma}
  \label{lem:aug-func-minimal-vc}
  For any $X \supset X_1$ let $H'_X = (V, E_{G'}[X])$ be a subgraph
  of $G'$. Suppose $\hC$ is a \emph{minimal} violated biset
  of $H'_X$ with respect to $h$. Then the following properties
  hold.
  \begin{itemize}
  \item $C' \subseteq X$.
  \item $G[C]$ is a connected subgraph of $G$.
  \end{itemize}
\end{lemma}
\begin{proof}
  For ease of notation we let $H'$ denote the graph $H'_X$.  Since
  $\hC$ is a violated biset in $H'$ we have $h(\hC) = 1$ and
  $|\delta_{H'}(\hC)| = 0$. From the definition of $h$, $\rV(\hC) = 2$
  and $\card{\bd(\hC)} +\card{\delta_{G[X]}(\hC)} = 1$.  Suppose there
  is a vertex $u \in C'$ such that $u \not \in X$. By Lemma~\ref{lem:aug-func-bd-vc}, $u\in C$. 
  First, $u$ is not a
  terminal since all terminals are in $X_1$ (and hence in $X$). Second
  $u$ is an isolated vertex in $H'$ since the only edges in $H'$ are
  between nodes in $X$. Consider the biset $\hC_1 = (C -u, C'-u)$
  obtained from $C$ by removing $u$.  Since $u$ is not a terminal we
  have $\rV(\hC) = \rV(\hC_1)$.  And since $u$ is isolated in $H'$ we
  have $\delta_{H'}(\hC_1) = \delta_{H'}(\hC)$, and
  $\bd(\hC_1) \subseteq \bd(\hC)$.  These facts imply that
  $h\hC_1) = 1$ and $\hC_1$ is a violated biset in $H'$. This
  contradicts minimality of $\hC$. Therefore $C' \subseteq X$.

  We now prove that $G[C]$ is connected. For sake of contradiction
  suppose it is not. Let $C_1, C_2$ be two non-empty sets that
  partition $C$ such that there is no edge between $C_1$ and $C_2$ in
  $G$; such a partition exists if $G[C]$ is not connected.  Note that
  $E_{H'}(C_1,C_2) = \emptyset$ since $H'$ is a subgraph of $G$.
  Define $\hC_1=(C_1, C_1 \cup \bd(\hC))$ and
  $\hC_2=(C_2, C_2 \cup \bd(\hC))$.  Since $\rV$ is bimaximal,
  $\rV(\hC) \leq \max\set{\rV(\hC_1), \rV(\hC_2)}$.  Thus, without loss of
  generality we can assume that $\rV(\hC_1) \geq \rV(\hC) \geq 2$.
  Since $E_{H'}(C_1,C_2) = \emptyset$ we have $\delta_{H'}(\hC_1)
  \subseteq \delta_{H'}(\hC)$. Since $\bd(\hC_1) = \bd(\hC)$ and
  $\hC$ was a violated biset it follows that $\hC_1$ is also
  a violated biset with respect to $h$ in $H'$. This contradicts
  the minimality of $\hC$.
\end{proof}

For any $X \supseteq X_1$ finding minimal violated bisets of $h$
with respect to the graph $H'_X$ can be easily done in polynomial time
via maxflow algorithms.

The function $h$ satisfies the same properties as those that arise in
the augmentation framework for \prob{Elem-SNDP} and hence we can apply
the primal-dual algorithm and analysis as captured by
Theorem~\ref{thm:main}. The algorithm outputs a node set $Q$ such that
$G[Q]$ covers $h$ and $w(Q \setminus X) = O(1)\opt$.  Since
$w(X) = O(1) \opt$ we have that $w(Q \cup X) = O(1) \opt$.  For planar
graphs we can obtain a concrete upper bound of $13 \opt$ using the
$3$-approximation for the first stage and a $10$-approximation for the
second stage.

\section{Concluding Remarks}
\label{sec:conclusions}
We obtained approximation algorithms for node-weighted network design
in planar and minor-closed families of graphs when the connectivity
requirements are larger than one. We built upon the insights from
\cite{DemaineHK14} as well as prior work via the augmentation
framework for \prob{SNDP}. The analysis of the primal-dual algorithm
that we present is probably not tight and it would be interesting to
obtain the tightest bound one can prove for the algorithm. It may be
possible to borrow ideas from \cite{BermanY12} and alter the algorithm
to obtain improved approximation ratios. For general \prob{VC-SNDP}
we obtain an improvement over the general graph case via our algorithm
for \prob{EC-SNDP} and a black-box reduction of \cite{ChuzhoyK12}.
For two important special cases of \prob{VC-SNDP}, namely
\prob{Rooted-VC-SNDP} and \prob{Subset-VC-SNDP},
$O(k \log k)$-approximations are known in the edge-weighted case
\cite{Nutov12,Bundit15} and the node-weighted case requires an
additional $O(\log n)$-factor. It would be interesting to
show that this additional factor is unnecessary
in planar graphs --- we note that the
results in \cite{Nutov12,Bundit15} are based on the augmentation
framework and hence some of our ideas  may be applicable.

Finally, it is an interesting question whether there is an $O(1)$-approximation 
for node-weighted \prob{EC-SNDP} and other network
design problems in planar graphs. Is the dependence on $k$ necessary?
Recall that for general graphs we expect a dependence on $k$ via the
hardness reduction from the \prob{$k$-Densest-Subgraph} problem
\cite{Nutov10}.  However, \prob{$k$-Densest-Subgraph} is
constant-factor approximable in planar graphs.

\bibliographystyle{plain}
\bibliography{nw-pc-sndp}

\begin{thebibliography}{10}

\bibitem{AgrawalKR95}
A.~Agrawal, P.~Klein, and R.~Ravi.
\newblock When trees collide: An approximation algorithm for the generalized
  steiner problem on networks.
\newblock {\em SIAM Journal on Computing}, 24(3):440--456, 1995.

\bibitem{BateniHL13}
M.~Bateni, M.~Hajiaghayi, and V.~Liaghat.
\newblock Improved approximation algorithms for (budgeted) node-weighted
  steiner problems.
\newblock In {\em International Colloquium on Automata, Languages, and
  Programming}, pages 81--92. Springer, 2013.

\bibitem{BateniHM11}
MohammadHossein Bateni, MohammadTaghi Hajiaghayi, and D{\'a}niel Marx.
\newblock Approximation schemes for steiner forest on planar graphs and graphs
  of bounded treewidth.
\newblock {\em Journal of the ACM (JACM)}, 58(5):21, 2011.

\bibitem{BermanY12}
P.~Berman and G.~Yaroslavtsev.
\newblock Primal-dual approximation algorithms for node-weighted network design
  in planar graphs.
\newblock In {\em Approximation, Randomization, and Combinatorial Optimization.
  Algorithms and Techniques}, pages 50--60. Springer, 2012.

\bibitem{BorradaileKM09}
Glencora Borradaile, Philip Klein, and Claire Mathieu.
\newblock An o(n log n) approximation scheme for steiner tree in planar graphs.
\newblock {\em ACM Trans. Algorithms}, 5(3):31:1--31:31, July 2009.

\bibitem{ByrkaGRS10}
J.~Byrka, F.~Grandoni, T.~Rothvo{\ss}, and L.~Sanit{\`a}.
\newblock An improved {LP}-based approximation for steiner tree.
\newblock In {\em Proceedings of the forty-second ACM symposium on Theory of
  computing}, pages 583--592. ACM, 2010.

\bibitem{ByrkaGRS13}
Jaros{\l}aw Byrka, Fabrizio Grandoni, Thomas Rothvoss, and Laura Sanit{\`a}.
\newblock Steiner tree approximation via iterative randomized rounding.
\newblock {\em Journal of the ACM (JACM)}, 60(1):6, 2013.

\bibitem{ChekuriEV12}
C.~Chekuri, A.~Ene, and A.~Vakilian.
\newblock Node-weighted network design in planar and minor-closed families of
  graphs.
\newblock In {\em Automata, Languages, and Programming}, pages 206--217.
  Springer, 2012.

\bibitem{ChekuriEV12-approx}
C.~Chekuri, A.~Ene, and A.~Vakilian.
\newblock Prize-collecting survivable network design in node-weighted graphs.
\newblock In {\em Approximation, Randomization, and Combinatorial Optimization.
  Algorithms and Techniques}, pages 98--109. Springer, 2012.

\bibitem{CheriyanVV06}
J.~Cheriyan, S.~Vempala, and A.~Vetta.
\newblock Network design via iterative rounding of setpair relaxations.
\newblock {\em Combinatorica}, 26(3):255--275, 2006.

\bibitem{ChuzhoyK12}
J.~Chuzhoy and S.~Khanna.
\newblock An ${O}(k^3 \log n)$-approximation algorithm for vertex-connectivity
  survivable network design.
\newblock {\em Theory of Computing}, 8:401--413, 2012.

\bibitem{DemaineHK14}
E.~D. Demaine, M.~Hajiaghayi, and P.~N. Klein.
\newblock Node-weighted {Steiner} tree and group {Steiner} tree in planar
  graphs.
\newblock {\em ACM Transactions on Algorithms (TALG)}, 10(3):13, 2014.

\bibitem{FleischerJW06}
L.~Fleischer, K.~Jain, and D.~P. Williamson.
\newblock Iterative rounding 2-approximation algorithms for minimum-cost vertex
  connectivity problems.
\newblock {\em Journal of Computer and System Sciences}, 72(5):838--867, 2006.

\bibitem{Fukunaga17}
Takuro Fukunaga.
\newblock Spider covers for prize-collecting network activation problem.
\newblock {\em ACM Transactions on Algorithms (TALG)}, 13(4):49, 2017.

\bibitem{GoemansGPSTW94}
M.~X. Goemans, A.~V. Goldberg, S.~Plotkin, D.~B. Shmoys, E.~Tardos, and D.~P.
  Williamson.
\newblock Improved approximation algorithms for network design problems.
\newblock In {\em Proceedings of the fifth annual ACM-SIAM symposium on
  Discrete algorithms}, pages 223--232, 1994.

\bibitem{GoemansW95}
M.~X. Goemans and D.~P. Williamson.
\newblock A general approximation technique for constrained forest problems.
\newblock {\em SIAM Journal on Computing}, 24(2):296--317, 1995.

\bibitem{GoemansW97}
M.~X. Goemans and D.~P. Williamson.
\newblock The primal-dual method for approximation algorithms and its
  application to network design problems.
\newblock {\em Approximation algorithms for NP-hard problems}, pages 144--191,
  1997.

\bibitem{GoemansW98}
Michel~X Goemans and David~P Williamson.
\newblock Primal-dual approximation algorithms for feedback problems in planar
  graphs.
\newblock {\em Combinatorica}, 18(1):37--59, 1998.

\bibitem{GuptaK11}
A.~Gupta and J.~K{\"o}nemann.
\newblock Approximation algorithms for network design: A survey.
\newblock {\em Surveys in Operations Research and Management Science},
  16(1):3--20, 2011.

\bibitem{Jain01}
K.~Jain.
\newblock A factor 2 approximation algorithm for the generalized {Steiner}
  network problem.
\newblock {\em Combinatorica}, 21(1):39--60, 2001.

\bibitem{JainMVW02}
K.~Jain, I.~I. Mandoiu, V.~V. Vazirani, and D.~P. Williamson.
\newblock A primal-dual schema based approximation algorithm for the element
  connectivity problem.
\newblock {\em J. Algorithms}, 45(1):1--15, 2002.

\bibitem{KleinR95}
P.~N. Klein and R.~Ravi.
\newblock A nearly best-possible approximation algorithm for node-weighted
  {Steiner} trees.
\newblock {\em J. Algorithms}, 19(1):104--115, 1995.

\bibitem{KortsarzN10}
G.~Kortsarz and Z.~Nutov.
\newblock Approximating minimum cost connectivity problems.
\newblock In {\em Dagstuhl Seminar Proceedings}. Schloss
  Dagstuhl-Leibniz-Zentrum f{\"u}r Informatik, 2010.

\bibitem{Kostochka84}
A.~V. Kostochka.
\newblock Lower bound of the hadwiger number of graphs by their average degree.
\newblock {\em Combinatorica}, 4(4):307--316, 1984.

\bibitem{Bundit15}
Bundit Laekhanukit.
\newblock An improved approximation algorithm for the minimum cost subset
  k-connected subgraph problem.
\newblock {\em Algorithmica}, 72(3):714--733, 2015.

\bibitem{Moldenhauer13}
C.~Moldenhauer.
\newblock Primal-dual approximation algorithms for node-weighted steiner forest
  on planar graphs.
\newblock {\em Information and Computation}, 222:293--306, 2013.

\bibitem{Nutov10}
Z.~Nutov.
\newblock Approximating {Steiner} networks with node-weights.
\newblock {\em SIAM Journal on Computing}, 39(7):3001--3022, 2010.

\bibitem{Nutov12}
Z.~Nutov.
\newblock Approximating minimum-cost connectivity problems via uncrossable
  bifamilies.
\newblock {\em ACM Transactions on Algorithms (TALG)}, 9(1):1, 2012.

\bibitem{Nutov12-latin}
Z.~Nutov.
\newblock Approximating steiner network activation problems.
\newblock In {\em Proc. of LATIN}, 2012.

\bibitem{Nutov18}
Zeev Nutov.
\newblock Erratum: Approximating minimum-cost connectivity problems via
  uncrossable bifamilies.
\newblock {\em ACM Trans. Algorithms}, 14(3):37:1--37:8, June 2018.

\bibitem{Panigrahi11}
D.~Panigrahi.
\newblock Survivable network design problems in wireless networks.
\newblock In {\em Proceedings of the twenty-second annual ACM-SIAM symposium on
  Discrete Algorithms}, pages 1014--1027. SIAM, 2011.

\bibitem{Vakilian13}
A.~Vakilian.
\newblock Node-weighted prize-collecting survivable network design problems.
\newblock Master's thesis, University of Illinois, Urbana-Champaign, 2013.

\bibitem{WilliamsonGMV95}
D.~P. Williamson, M.~X. Goemans, M.~Mihail, and V.~V. Vazirani.
\newblock A primal-dual approximation algorithm for generalized {Steiner}
  network problems.
\newblock {\em Combinatorica}, 15(3):435--454, 1995.

\end{thebibliography}
\appendix
\section{Omitted Proofs}
\subsection{Omitted Proofs of Section~\ref{sec:prelim}}
\label{subsec:omitted-preliminaries}

\begin{proofof}{Proposition~\ref{prop:boundary-bisubmodular}}
The proof follows from the following propositions.
\end{proofof}
\begin{prop}
For any two bisets $\hS$ and $\hT$,
$\card{\bd(\hS)} + \card{\bd(\hT)} \geq \card{\bd(\hS\cap\hT)} + \card{\bd(\hS\cup\hT)}.$
\end{prop}
\begin{proof}
Consider a vertex $v$ that contributes to RHS. Then one of the following cases holds:
\begin{itemize}
	\vspace{-.1in}
	\item {$v \in \bd(\hS\cap\hT) \setminus \bd(\hS\cup\hT)$.}
	This implies that $v\in S' \cap T'$ and at least one of $S$
	and $T$ does not contain $v$. Thus $v\in \bd(\hS) \cup \bd(\hT)$.

	\vspace{-.1in}
	\item {$v \in \bd(\hS\cup\hT)\setminus \bd(\hS\cap\hT)$.}
	This implies that $v\notin S\cup T$ and at least one of $S'$
	and $T'$ contains $v$. Thus $v\in \bd(\hS) \cup \bd(\hT)$.

	\vspace{-.1in}
	\item {$v \in \bd(\hS\cap\hT) \cap \bd(\hS\cup\hT)$.}
	This implies that $v\in S'\cap T'$ and $v\notin S\cup T$. Thus $v\in\bd(\hS) \cap \bd(\hT)$.
\end{itemize}
Hence the contribution of each vertex to LHS is not less than its contribution to RHS and the statement holds.
\end{proof}

\begin{prop}
For any two bisets $\hS$ and $\hT$,
$\card{\bd(\hS)} + \card{\bd(\hT)} \geq \card{\bd(\hS\setminus\hT)} + \card{\bd(\hT\setminus\hS)}.$
\end{prop}
\begin{proof}
Consider a vertex $v$ that contributes to RHS. Then it would be one the following cases:
\begin{itemize}
	\vspace{-.1in}
	\item {$v \in \bd(\hS\setminus\hT) \setminus \bd(\hT\setminus\hS)$.}
	This implies that $v\in S'\setminus T$. If we have also $v\notin S$ then
	$v\in \bd(\hS)$. Otherwise; since $v\notin (S\setminus T')$, we must
	have $v\in T'$. Thus $v\in \bd(\hT)$.

	\vspace{-.1in}
	\item {$v \in \bd(\hT\setminus \hS) \setminus \bd(\hS\setminus\hT)$.}
	Similar to the above case either $v\in \bd(\hS)$ or $v\in \bd(\hT)$.

	\vspace{-.1in}
	\item {$v \in \bd(\hT\setminus \hS) \cap \bd(\hS\setminus\hT)$.}
	This implies that $v\in T'\setminus S$ and $v\in S'\setminus T$.
	Thus $v\in \bd(\hS) \cap \bd(\hT)$.
\end{itemize}
Hence the contribution of a vertex to LHS is at least the contribution of the vertex to RHS and the statement holds.
\end{proof}

\subsection{Proof of Lemma~\ref{lem:pair-laminar-witness-family}}
Our first observation is that, if $M$ is a set of non-redundant edges, we
can pick a witness biset for each edge of $M$.





\begin{lemma} [Lemma $4.1$ in~\cite{FleischerJW06}] \label{lem:pair-uncrossing}
	Let $\hS_{e_1}, \hS_{e_2}$ be $F$-witness bisets for $e_1$
	and $e_2$, respectively. Then one of the following holds:
	\begin{enumerate}[(i)]
		\item The bisets $\hS_{e_1} \cap \hS_{e_2}$ and $\hS_{e_1}
		\cup \hS_{e_2}$ are $F$-witness bisets for distinct edges in
		$\set{e_1, e_2}$.
		\item The bisets $\hS_{e_1} \setminus \hS_{e_2}$ and $\hS_{e_2} \setminus
		\hS_{e_1}$ are $F$-witness bisets for distinct edges in
		$\set{e_1, e_2}$.
	\end{enumerate}
\end{lemma}

Using the following lemma, we can show that, if we replace two
overlapping witness bisets with the witness bisets guaranteed by
Lemma~\ref{lem:pair-uncrossing}, the number of pairs of overlapping bisets
decreases and thus we are making progress towards a laminar family of witness bisets.

\begin{lemma}[Lemma 4.3 in~\cite{FleischerJW06}] \label{lem:pair-uncrossing-progress}
	Let $\hS_1$, $\hS_2$ be two overlapping bisets and $\hT$ be an arbitrary biset.
	Then the number of
	pairs of $\set{(\hS_1, \hS_2), (\hS_1, \hT), (\hS_2, \hT)}$ that overlap
	is at least the number of pairs of $\set{(\hS_1 \cap \hS_2, \hS_1 \cup
	\hS_2), (\hS_1 \cap \hS_2, \hT), (\hS_1 \cup \hS_2, \hT)}$ that overlap.
	Similarly, the number of pairs of $\set{(\hS_1, \hS_2), (\hS_1,
	\hT), (\hS_2, \hT)}$ that overlap is at least the number of pairs
	of $\set{(\hS_1 \setminus \hS_2, \hS_2 \setminus \hS_1), (\hS_1 \setminus \hS_2, \hT), (\hS_2 \setminus \hS_1,
	\hT)}$ that overlap.
\end{lemma}

\begin{proofof}{Lemma~\ref{lem:pair-laminar-witness-family}}
	Let $\mF$ be the initial family of $F$-witness bisets of edges in
	$M$. If no two bisets in $\mF$ overlap, $\mF$ is the desired
	family. Otherwise, let $\hS_{e_1}$ and $\hS_{e_2}$ be two
	bisets that overlap. By Lemma~\ref{lem:pair-uncrossing} and
	Lemma~\ref{lem:pair-uncrossing-progress}, we can replace
	$\hS_{e_1}$ and $\hS_{e_2}$ with either $\hS_{e_1} \cap \hS_{e_2}$
	and $\hS_{e_1} \cup \hS_{e_2}$ or by $\hS_{e_1}\setminus \hS_{e_2}$ and
	$\hS_{e_2} \setminus \hS_{e_1}$. The resulting family is an $F$-witness family
	of $M$ that has fewer overlapping bisets. Thus we can repeat this
	process until we get a non-overlapping $F$-witness family.
\end{proofof}

\end{document}